\documentclass[journal]{IEEEtran}
\IEEEoverridecommandlockouts

\usepackage{amsmath,amssymb,amsfonts,amsthm}
\usepackage{inputenc}
\usepackage{enumerate}
\usepackage{graphicx}
\usepackage{multicol}
\usepackage{multirow}
\newtheorem{thm}{Theorem}

\newtheorem{lem}{Lemma}

\newtheorem{const}{Construction}

\newtheorem{exam}{Example}

\newtheorem{rem}{Remark}

\def\0{{\mathbf 0}}

\newcommand{\F}{\mathbb{F}}
\newcommand{\Z}{\mathbb{Z}}

\begin{document}

\sloppy

%% Paper Title
%% You can use linebreaks \\ within to get better formatting as
%% desired. 
\title{New families of quantum stabilizer codes from Hermitian self-orthogonal algebraic geometry codes}

%% Author names and affiliations:
%%
%% Avoiding spaces at the end of the author lines is not a problem with
%% conference papers because we don't use \thanks or \IEEEmembership.
%%
%% For several authors with only one affiliation:
%%
% \author{
%   \IEEEauthorblockN{Hui-Ting Chang and Stefan M.~Moser}
%   \IEEEauthorblockA{Department of Electrical and Computer Engineering\\
%     National Chiao Tung University (NCTU)\\
%     Hsinchu, Taiwan\\
%     Email: \{email-of-hui-ting,email-of-stefan\}@ieee.org} 
% }
%%
%% For up to three affiliations:
%%
\author
{
{
Lin Sok \thanks{This research work is supported by Anhui Provincial Natural Science Foundation with grant number 1908085MA04.

Lin Sok is with School of Mathematical Sciences, Anhui University,  230601 Anhui,  P. R. China
$\&$ Department of Mathematics, Royal University of Phnom Penh, 12156 Phnom Penh
(email: soklin\_heng@yahoo.com). }
}
%\author{Anonymous submission to DSD}
%School of Mathematical Sciences, Anhui University,  230601 Anhui,  P. R. China
}

\date{}
%%
%% For over three affiliations, or if they all won't fit within the width
%% of the page, use this alternative format:
%%
% \author{
%   \IEEEauthorblockN{
%     Michael Shell\IEEEauthorrefmark{1},
%     Homer Simpson\IEEEauthorrefmark{2},
%     James Kirk\IEEEauthorrefmark{3}, 
%     Montgomery Scott\IEEEauthorrefmark{3} and
%     Eldon Tyrell\IEEEauthorrefmark{4}}
%   \IEEEauthorblockA{
%     \IEEEauthorrefmark{1}School of Electrical and Computer Engineering\\
%     Georgia Institute of Technology, Atlanta, Georgia 30332--0250\\ 
%     Email: see http://www.michaelshell.org/contact.html}
%   \IEEEauthorblockA{
%     \IEEEauthorrefmark{2}Twentieth Century Fox, Springfield, USA\\
%     Email: homer@thesimpsons.com}
%   \IEEEauthorblockA{
%     \IEEEauthorrefmark{3}Starfleet Academy, San Francisco, California 96678-2391\\
%     Telephone: (800) 555--1212, Fax: (888) 555--1212}
%   \IEEEauthorblockA{
%     \IEEEauthorrefmark{4}Tyrell Inc., 123 Replicant Street, Los Angeles, California 90210--4321}
% }

%% Use for special paper notices

%% To balance the two columns, you should reduce the text-height of
%% the last page using the following command:
%%%%%%%%%%%%%%%%%%%%%%%%%%%%%%%%%%%%%%%%%%%%%%%%%%%%%%%%%%%%%%%%%%%%%
%%\addtolength{\textheight}{-9.35cm}
%%%%%%%%%%%%%%%%%%%%%%%%%%%%%%%%%%%%%%%%%%%%%%%%%%%%%%%%%%%%%%%%%%%%%
%% with an appropriate value. This command must be place on the second
%% last page, i.e., for a one-page abstract here, for a two-page
%% abstract right after the \maketitle command.

%% Create the title:
\maketitle

%% Abstract: 
%% For the final version of the accepted paper, please make sure you
%% remove the comment "THIS PAPER IS ELIGIBLE FOR THE STUDENT PAPER
%% AWARD."
%%

\begin{abstract}
There has been a lot of effort to construct good quantum codes from the classical error correcting codes. Constructing new quantum codes, using Hermitian self-orthogonal codes, seems to be a difficult problem in general. In this paper, Hermitian self-orthogonal codes are studied from algebraic function fields. Sufficient conditions for the Hermitian self-orthogonality of an algebraic geometry code are presented. New Hermitian self-orthogonal codes are constructed from projective lines, elliptic curves, hyper-elliptic curves, Hermitian curves, and Artin-Schreier curves. In addition, over the projective lines, we construct new families of MDS quantum codes with parameters $[[N,N-2K,K+1]]_q$ under the following conditions: 
i) $N=t(q-1)+1$ or $t(q-1)+2$ with $t|(q+1)$ and $K=\lfloor\frac{t(q-1)+1}{2t}\rfloor+1$; 
ii) $(n-1)|(q^2-1)$, $N=n$ or $N=n+1$, $K_0=\lfloor\frac{n+q-1}{q+1}\rfloor$, and $K\ge K_0+1$; 
%and $K\ge K_0+2$ if $(n-1){\not|}K_0(q+1)$, $(n-1){\not|}(K_0+1)(q+1)$ and $K_0\le \lfloor\frac{n}{4}\rfloor-1$;  
iii) $N=tq+1$, $\forall~1\le t\le q$ and $K=\lfloor\frac{tq+q-1}{q+1}\rfloor+1$; 
iv) $n|(q^2-1)$, $n_2=\frac{n}{\gcd (n,q+1)}$, $\forall~ 1\le t\le \frac{q-1}{n_2}-1$, $N=(t+1)n+2$ and $K=\lfloor \frac{(t+1)n+1+q-1}{q+1}\rfloor+1$.

\end{abstract}

\section{Introduction}
%%%%%%%%%%%%%%%%%%%%%%%%%%%%%%%
Quantum information and quantum computation are one of the hot research topics. Quantum codes have been of great interest to many researchers after the earliest work \cite{Cal Rai} followed by \cite{Ash Kni,Ket Kla}. In \cite{Ash Kni,Ket Kla}, they introduced some method to derive non-binary quantum codes from the classical error correcting codes. More specifically, their constructions use some properties of Euclidean, Hermitian, and symplectic self-orthogonality. There are two well-known types of construction of quantum codes from the classical ones: the CSS construction which employs Euclidean self-orthogonality, and the other one which uses Hermitian self-orthogonality.

Algebraic geometry codes were discovered in 1980 by Goppa. Goppa showed in his paper \cite{Goppa} how to construct linear codes from algebraic curves over a finite field. The parameters of an algebraic geometry code can be computed via the degree of the divisors associated to the code. It is well known that algebraic geometry (AG) codes have asymptotically good parameters, and it was the first time that linear codes improved the so-called Gilbert-Vasharmov bound. Algebraic geometry codes are one of the good candidates for constructing self-orthogonal codes with respect to both Euclidean and Hermitian inner products. Euclidean self-orthogonal AG codes were studied by Stichtenoth \cite{Stich88} and Driencourt {\em et al.} \cite{DriSti}, where they determined the Euclidean dual and also characterized such codes. To the best of our knowledge, there is no formula to calculate the Hermitian dual of an AG code as in the Euclidean case, and thus constructing Hermitian self-orthogonal AG codes is a difficult problem in general.

Quantum MDS codes form an optimal family of quantum codes. It is very well known that any $q$-ary quantum stabilizer code, derived from an MDS $q$-ary linear code, is again an MDS code. It should be noted that the length of a $q$-ary quantum stabilizer code can not exceed $q^2+1$ if the classical MDS conjecture holds, and thus $q$-ary quantum stabilizer codes can exist in very restricted conditions on their lengths and their dimensions. Reed-Solomon codes are a class of MDS linear codes, and most of the known MDS quantum stabilizer codes were constructed, using these codes. Constructing $q$-ary (MDS) quantum stabilizer codes from Hermitian self-orthogonal codes has been a challenge problem in the literature since the methods, developed for Euclidean case, may not work anymore. The fact is that in the Hermitian inner product, one has to consider the $q$-th power in the summand. To construct new MDS Hermitian self-orthogonal codes over $\F_{q^2}$ from the generalized Reed-Solomon codes, the authors \cite{JinXing14} considered a system of homogenous equations over $\F_{q^2}$ with some restricted solutions in $\F_q$ and found some suitable evaluation points.

The problem of constructing quantum MDS codes with length $n\le q+1$ was completely solved in \cite{GraBet,Gra Bet2}. Other known families of $q$-ary MDS quantum stabilizer codes with parameters $[[n,n-2d+2,d]]$ can be summarized in Table \ref{table:known}.\\
\begin{table}[tth]
\caption{Some known results}
\label{table:known}
\begin{enumerate}
%\begin{alignment}
 \item $n=q^2+1$ and $d=q+1$ (see \cite{Li Xin Wan2})\\
\item $n=q^2+1$ and $ d\le q+1$ for even $q$ and odd $d$  (see \cite{Gu11})\\
\item $n=q^2+1$ and $d\le q+1$ for $q\equiv 1\mod{4}$ and even $d$ (see \cite{KZ12})\\
\item $n= q^2$ and $d\leq q$ (see \cite{Gra Bet2} and \cite{Li Xin Wan2,LinLingLuoXin})\\
\item $n=(q^2+1)/2$ and $q/2+1< d\le q$ for odd $q$ (see \cite{KZ12})\\
\item $n=q^2-s$ and $q/2+1<d\le q-s$ for $0\le s<q/2-1$,\\
$n=(q^2+1)/2-s$ and $q/2+1<d\le q-s$ for $0\le s<q/2-1$ (by 2), 3) and the puncturing rule \cite{GraRot}, see also \cite{FLX06});\\
\item $n=q^2+1$ and $d\le q-1$ or $d=q+1$,\\
$n=r(q-1)+1$ and $d\le (q+r+1)/2$ for $q\equiv r-1\mod{2r}$,\\
$n=(q^2+2)/3$ for $3|(q+1)$, $d\le (2q+2)/3$ (see \cite{JinXing14})\\
\item $n=tq$, $1\le t\le q$ and $2\le d\le \lfloor\frac{tq+q-1}{q+1} \rfloor +1$,\\
$n=t(q+1)+2$, $1\le t\le q-1$ and $2\le d\le t +2$, $(p,t,d)\not=(2,q-1,q)$ (See \cite{FangFu})
\end{enumerate}
%\end{alignment}
\end{table}
In this work, we provide sufficient conditions for an AG code to be Hermitian self-orthogonal. The critical problem in our sufficient conditions is to choose some sets of evaluation points with some desired properties. Such sets can be obtained from some special subsets of $\F_{q^2}$ as well as from some multiplicative subgroups of $\F_{q^2}^*$ and their cosets. We explore, via such sets, the construction of MDS Hermitian self-orthogonal codes from projective lines, and then the non-MDS codes from other algebraic curves. Over the projective lines, we present a method for embedding an MDS Hermitian self-orthogonal AG code of dimension $k$ into an MDS Hermitian self-orthogonal code of dimension $ k+1$. For some special length $n$ such that $(n-1)|(q^2-1)$, we provide a recursive construction of an MDS Hermitian self-orthogonal code with dimension $k'+1$, from an MDS Hermitian self-orthogonal code of dimension $k'$, under the very simple conditions of divisibility $(n-1){|}k'(q+1)$ or $(n-1){\not|}k'(q+1)$ with $k'$ satisfying $2k'+q\le n$. We obtain several new families of both MDS and non-MDS $q$-ary quantum codes with good parameters, for instance, some new $q$-ary MDS quantum codes over 
$\F_q, q=13,17,19,23,25,27$, with parameters 
$[[25,11,8]]_{13}$,
$[[33,15,10]]_{17}$,
$[[38,18,11]]_{19}$,
$[[46,22,13]]_{23}$,
$[[49,23,14]]_{25}$,
$[[54,26,15]]_{27}$
as well as some new codes with parameters 
$[[24,18,3]]_4$,
%$[[20,12,4]]_4$,
$[[80,64,8]]_8$,
$[[95,89,3]]_5$, 
$[[95,87,3]]_5$, 
$[[95,83,4]]_5$, 
${ [[91,81,4]]_7}$, 
${ [[176,168,3]]_8}$,
${ [[369,361,3]]_9}$, where all of these parameters improve those in the database \cite{Data-BieEde}. Among the families of MDS quantum codes obtained, one family contains $q$-ary quantum codes with large minimum distances $d=\lfloor \frac{n}{2t}\rfloor+2$ with $t|(q+1)$, where $n=t(q-1)+1$ or $n=t(q-1)+2$ is the code length. 
Our quantum codes, constructed from elliptic curves, improve the code lengths from $q+\lfloor \sqrt{2q}\rfloor -5 $ \cite{JinXing14} to $q^2+ {2q}$. From hyper-elliptic curves, we obtain a family of codes with intermediate lengths which are new compared to \cite{Jin}. 
Our codes from Artin-Schreier curves are relatively new since they have not been considered elsewhere in the literature.

We summarize the contribution of our work in Theorem \ref{thm:Q-all}.
\begin{thm} \label{thm:Q-all} The exsitence of quantum codes is given as follows.
\begin{enumerate}
\item Let $q=p^m$, $N=t(q-1)+1$ with $t|(q+1)$, and $k=\lfloor \frac{N}{2t}\rfloor$. Then
\begin{enumerate}
\item there exists an MDS quantum code with parameters $[[N+1,N-2k-1,k+2]]_q$ if $(N-1)|k(q+1)$;
\item there exists an MDS quantum code with parameters $[[N,N-2k-1,k+2]]_q$ if $(N-1){\not|}k(q+1)$.
\end{enumerate}

\item Let $q=p^m$ and $1\le k\le \lfloor \frac{N+q-1}{q+1} \rfloor$. Assume that one of the following condition holds:
\begin{enumerate}[i)]
\item $(N-1)|(q^2-1)$, 
\item $N=(t+1)n+1$, $n|(q^2-1)$, $n_2=\frac{n}{\gcd (n,q+1)}$, $1\le t\le \frac{q-1}{n_2}-1$,
\item $N=tq$, $1\le t\le q$.
\end{enumerate} 
Then
\begin{enumerate}
\item there exists an MDS quantum code with parameters $[[N,N-2k,k+1]]_q$;
\item there exists an MDS quantum code with parameters $[[N+1,N-2k-1,k+2]]_q$ if $k(q+1)=N-1$;
\item there exists an MDS quantum code with parameters $[[N,N-2k,k+2]]_q$ if $N-1+q\not=k(q+1)\not=N-1$.
\end{enumerate}
Moreover, for  $k_0=\lfloor\frac{n+q-1}{q+1}\rfloor$ and $(n-1)|(q^2-1)$, there exists an MDS quantum code with parameters $[[n',n'-2k,k+1]]_q$, where $n'=n$ or $n'=n+1$ and $k\ge k_0+2$, if $(n-1){\not|}k_0(q+1)$, $(n-1){\not|}(k_0+1)(q+1)$ and $2k_0+q\le n$.

\item Let $q=2^s$ and $2\le k\le \lfloor \frac{2n+1+q}{q+1} \rfloor$. Then
there exists a quantum code with parameters $[2n,2n-2k+2,k-1]_{q}$  if {\bf Assumption 1} holds.
\item Let $q=2^m, m\ge 2$ and $1+q/2\le k\le \lfloor \frac{2N+2q-1}{q+1} \rfloor$. Then
\begin{enumerate}
\item for $N=q^2$, there exists a quantum code with parameters $[[2N,2N-2k+q,\ge k-q+1]]_{q}$;
\item for $(N-1)|(q^2-1)$, there exists a quantum code with parameters $[[2N,2N-2k+q,\ge k-q+1]]_{q}$.
\end{enumerate}

\item 
 Let $q=p^m\ge 4$ and $1+q(q-1)/2\le k\le \lfloor \frac{Nq+q^2-1}{q+1} \rfloor$. Then, there exists a quantum code with parameters $[[Nq,Nq-2k+q(q-1), \ge k+1-q(q-1) ]]_{q}$ if one of the following conditions holds:
\begin{enumerate}[i)]
\item $(N-1)|(q^2-1)$, 
\item $N=(t+1)n+1$, $n|(q^2-1)$, $n_2=\frac{n}{\gcd (n,q+1)}$, $1\le t\le \frac{q-1}{n_2}-1$,
\item $N=tq$, $1\le t\le q$.
\end{enumerate} 

\item  Let $q=p^m\ge 4$ and $1+\frac{(q-1)^2}{4}\le k\le \lfloor \frac{s+(q-1)^2/2+q-1}{q+1} \rfloor$. Then, there exists a quantum code with parameters $[[s,s-2k+\frac{(q-1)^2}{2}, \ge k+1-\frac{(q-1)^2}{2} ]]_{q}$.
\item Let $q=p^m\ge 4$ and $t$ be a positive integer such that $\gcd(t,(q+1))|\frac{q+1}{2}$. Put $s=\gcd(t(q-1),(q+1)(q-1))$ and $1+(t-1)(q-1)/2\le k\le \lfloor \frac{(s+1)q+t(q-1)}{q+1} \rfloor$. Then, there exists a quantum code with parameters $[[(s+1)q,(s+1)q-2k+(t-1)(q-1), \ge k+1-(t-1)(q-1) ]]_{q}$ if one of the following condition holds
\begin{enumerate}[i)]
\item $q$ being odd and $t$ being any positive integer;
\item $q$ being even and $t$ being an odd positive integer.
\end{enumerate} 

\end{enumerate}
\end{thm}

The paper is organized as follows: Section \ref{section:pre} gives preliminaries and background on algebraic geometry codes. Section \ref{section:constructions} provides sufficient conditions for a special AG code to be Hermitian self-orthogonal and gives construction methods for such codes. Hermitian self-orthogonal codes are constructed from projective lines, elliptic curves, hyper-elliptic curves, Hermitian curves, and Artin-Schreier curves. We give an application to construct quantum stabilizer codes in Section \ref{section:application}.

\section{Preliminaries}\label{section:pre}

The finite field with $q$ elements is denoted by $\F_q$. We write parameters $[n,k,d]_q$ for a linear code  over ${{\mathbb F}_q}$ of length $n$, dimension $k$, and minimum distance $d$. For an $[n,k,d]_q$ linear code, the Singleton bound gives the constraint among the code parameters as follows: $$d\le n-k+1.$$
A code meeting the above bound is called {\em Maximum Distance Separable} ({MDS}). A code is called {\em almost} MDS if its minimum distance is one unit less than the MDS case. 
A code is called {\it optimal} if it has the highest possible minimum distance for its length and dimension.

The {\em Euclidean} (resp. {\em Hermitian}) inner product of ${\bf{x}}=(x_1,
\dots, x_n)$ and ${\bf{y}}=(y_1, \dots, y_n)$ in ${\mathbb F}_q^n$ (resp. in $\F_{q^2}^n$) is defined by
$$<{\bf{x}},{\bf{y}}>_E=\sum_{i=1}^n x_i y_i~(\text{resp.} \\
<{\bf{x}},{\bf{y}}>_H=\sum_{i=1}^n x_i y_i^q).\\
$$ The Euclidean (resp. Hermitian) {\em dual} of $C$,
denoted by $C^{\perp}$ (resp. $C^{\perp_H}$), is the set of vectors orthogonal to every
codeword of $C$ under the Euclidean (resp. Hermitian) inner product. 
A linear code $C$ is called Euclidean (resp. Hermitian) self-orthogonal if $C\subseteq C^\perp$ (resp. $C\subseteq C^{\perp_H}$). For ${\bf x}=(x_1,\hdots,x_n)\in \F_q^n$, we write ${\bf x}^q:=(x_1^q,\hdots,x_n^q)$, and denote $C^q:=\{{\bf c}^q|{\bf c}\in C\}.$
Then it is easy to see that 
\begin{equation}\label{eq:E-H-dual}
C^{\perp_H}=(C^q)^\perp=\left(C^\perp\right)^q.
\end{equation}
For a linear code $C\subseteq \F_q^n$ and ${\bf a}=(a_1,\hdots,a_n)\in (\F_q^*)^n$, we define

\begin{equation}\label{eq:equivalent-code}
{\bf a}\cdot C:=\{{\bf a} \cdot {\bf c}| {\bf c}\in C\},
\end{equation}
where ${\bf a\cdot c}=(a_1c_1,\hdots,a_nc_n)$.
It can be easily checked that ${\bf a}\cdot C$ is a linear code if and only if $C$ is a linear code. Moreover, the codes $C$ and ${\bf a}\cdot C$ have the same dimension, minimum Hamming distance, and weight distribution.

We refer to Stichtenoth \cite{Stich} for undefined terms related to algebraic function fields. 

Let ${\cal X}$ be a smooth projective curve of genus $g$ over $\F_q.$
The field of rational functions of ${\cal X}$ is denoted by $\F_q({\cal X}).$ Function fields of
algebraic curves over a finite field can be characterized as
finite separable extensions of $\F_q(x)$. We identify points on the curve ${\cal X}$ with places of the
function field $\F_q({\cal X}).$ A point on ${\cal X}$ is called rational if all of
its coordinates belong to $\F_q.$ Rational points can be identified
with places of degree one. We denote the set of $\F_q$-rational
points of ${\cal X}$ by  ${\cal X}(\F_q)$.

A divisor $G$ on the curve ${\cal X}$ is a formal sum $\sum\limits_{P\in {\cal X}}n_PP$ with only finitely many non-zeros $n_P\in \Z$.
The support of $G$ is defined as $supp(G):=\{P|n_P\not=0\}$. The degree of $G$ is defined by $\deg(G):=\sum\limits_{P\in {\cal X}}n_P\deg(P)$. 
For two divisors $G=\sum\limits_{P\in {\cal X}}n_PP$ and $H=\sum\limits_{P\in {\cal X}}m_PP$,  we say that $G\le H$ if $n_P\le m_P$ for all places $P\in {\cal X}$. 

For a nonzero rational function $f$ on the curve $\cal X$, we define the {\em principal} divisor of $f$ as
$$(f):=\sum\limits_{P\in {\cal X}}v_P(f)P,$$ 
where $v_P$ denotes the discrete valuation map \cite[Definition 1.1.9]{Stich}.

If $Z(f)$ (resp. $N(f)$) denotes the set of zeros (resp. poles) of $f$, we define the zero divisor and pole divisor of $f$, respectively by

$$
\begin{array}{c}
(f)_0:=\sum\limits_{P\in Z(f)}v_{P}(f)P,\\
(f)_\infty:=\sum\limits_{P\in N(f)}-v_{P}(f)P.\\
\end{array}
$$
Then, $(f)=(f)_0-(f)_\infty $, and it is well known that the principal divisor $(f)$ has degree $0.$

For a divisor $G$ on the curve $\cal X$, we define
$${\cal L}(G):=\{f\in \F_q({\cal X})\backslash \{0\}|(f)+G\ge 0\}\cup \{0\},$$ 
and 
$${\Omega}(G):=\{\omega\in \Omega\backslash \{0\}|(\omega)-G\ge 0\}\cup \{0\},$$
where $\Omega:=\{fdx|f\in \F_q({\cal X})\}$, the set of differential forms on $\cal X$. It is well known that, for a differential form $\omega$ on $\cal X$, there exists a unique a rational function $f$ on $\cal X$ such that $$\omega=fdt,$$
where $t$ is a local uniformizing parameter. In this case, we define the divisor associated to $\omega$ by $$(\omega)=\sum\limits_{P\in {\cal X}}v_P(\omega)P, $$
where $v_P(\omega):=v_P(f).$

For $n$ pairwise distinct rational points 
$P_1, \ldots, P_n$ on ${\cal X}$ and for two disjoint divisors $D=P_1+\cdots+P_n$ and $G$, 
the algebraic geometry code $C_{\cal L}(D,G)$ and the differential algebraic geometry code $C_\Omega(D,G)$ are defined as the images of the linear maps 
\begin{align*}
ev_{\cal L} :~ &{\cal L}(G)~\longrightarrow~\F_q^n, ~f \mapsto (f(P_1), \ldots, f(P_n) ),\\
ev _{\Omega} :~ &\Omega(G-D)~\longrightarrow~\F_q^n, 
                  \omega \mapsto (\text{Res}_{P_1}(\omega), \ldots, \text{Res}_{P_n}(\omega)),
\end{align*}
respectively, where $\text{Res}_{P}(\omega)$ denotes the residue of $\omega$ at point $P.$

The parameters of an algebraic geometry code $C_{\cal L}(D,G)$ are given as follows.
\begin{thm}\textnormal{\cite[Corollary 2.2.3]{Stich}}\label{thm:distance} Assume that $2g -2 < deg(G) < n.$ Then, the code $C_{\cal L}(D,G)$  has dimension and minimum distance satisfying
\begin{equation}
k=\deg (G)-g+1\text{ and } d\ge n-\deg (G).
\label{eq:distance}
\end{equation}

\end{thm}

 For ${\bf a}=(\alpha_1,\hdots,\alpha_n),{\bf v}=(v_1,\hdots,v_n)\in \F_q^n$ such that $\alpha_1,\hdots,\alpha_n$ are all distinct, and $v_1,\hdots,v_n$ are all nonzeros, it is well known that the generalized Reed-Solomon code, defined by 

$$
\begin{array}{ll}
GRS_k({\bf a},{\bf v}):=\{(v_1f(\alpha_1),\hdots, v_nf(\alpha_n))| &f(x)\in \F_q(x),\\ &\deg{f}\le k-1\},\\
\end{array}
$$

 is an MDS code. Furthermore, it is shown 
in \cite[Proposition 2.3.3]{Stich}, that any algebraic geometry code $C_{\cal L}(D,G)$ with $\deg (G)=k-1$ is equal to the generalized Reed-Solomon code $GRS_k({\bf a},{\bf v})$ defined above. Moreover, their parameters are related as follows. For all $1\le i \le n,$\\

$
\begin{cases}
\alpha_i=x(P_i),\\
v_i=u(P_i)
\text{ for some $u(x)\in \F_q(x)$ satisfying }\\
  (u)=(k-1)P_{\infty}-G .\\
\end{cases}
$\\

For $0\le j\le k-1$, the vectors
$$(ux^j(P_1),\hdots,ux^j(P_{n}))=(v_1\alpha_1^j,\hdots,v_{n}\alpha_{n}^j)$$ constitute a basis of $C_{\cal L}(D,G)$, and thus a generator matrix of $C_{\cal L}(D,G)$ can be expressed as

\begin{equation*}
\left(
\begin{array}{cccc}
v_1&v_2&\hdots&v_{n}\\
v_1\alpha_1&v_2\alpha_2&\cdots&v_{n}\alpha_{n}\\
\vdots&\vdots&\cdots&\vdots\\
v_1\alpha_1^{k-2}&v_2\alpha_2^{k-2}&\cdots&v_{n}\alpha_{n}^{k-2}\\
v_1\alpha_1^{k-1}&v_2\alpha_2^{k-1}&\cdots&v_{n}\alpha_{n}^{k-1}\\
\end{array}
\right).
\end{equation*}

The dual of the algebraic geometry code $C_{\cal L}(D,G)$ can be described \cite[Theorem 2.2.8]{Stich} as follows.

\begin{lem}\label{lem:01}With the above notation, the Euclidean dual of $C_{\cal L}(D,G)$ is $C_{\Omega}(D,G)$.
\end{lem}

Moreover, the differential code $C_{\Omega}(D,G)$ is determined in \cite[Proposition 2.2.10]{Stich} as follows.
\begin{lem}\label{lem:1}With the above notation, 
$C_{\Omega}(D,G)={\bf e}\cdot C_{\cal L}(D,D-G+(\omega))$ for some differential function $\omega$ satisfying $v_{P_i}(\omega)=-1$ for $1\le i \le n$ and ${\bf e}=(\text{Res}_{P_i}(\omega),\hdots, \text{Res}_{P_n}(\omega))$. 
\end{lem}

The Euclidean self-orthogonality of an algebraic geometry code can be characterized as follows.

\begin{lem}\textnormal{\cite[Corollary 2.2.11]{Stich}}\label{lem:char1} With the above notation, assume that there exists a differential form $\omega $ satisfying
\begin{enumerate}
\item  $v_{P_i}(\omega)=-1, 1\le i \le n $ and 
\item $\text{Res}_{P_i}(\omega)=a_i^2$, $1\le i \le n,$  for some $a_i\in {\F^*_q}.$
\end{enumerate} 
If $2G\le D+(\omega),$ then there exists a divisor $G'$ such that $C_{\cal L}(D,G)$ is  equivalent to $C_{\cal L}(D,G')$, and  $C_{\cal L}(D,G')$ is Euclidean self-orthogonal.
\end{lem}

\section{Construction of Hermitian self-orthogonal codes}\label{section:constructions}

In the sequel, for two divisors $D=P_1+\cdots+P_n$ and $G$ such that $\text{supp}(D)\cap \text{supp}(G)=\emptyset$, we denote
\begin{equation}
\begin{array}{ll}
{\cal L}_q(D,G):=\{f\in \F_{q^2}({\cal X})\backslash \{0\}|&f(P_i)\in \F_q\text{ for }1\le i\le n,\\& (f)+G\ge 0\}\cup \{0\}.
\end{array}
\end{equation}
For ${\bf v}=(v_1,\hdots,v_n)$ with $v_i\in \F_{q^2}^*,$ define the following algebraic geometry codes
\begin{equation}
C_{{\cal L}_{q}}(D,G):=\{(f(P_{1}),\hdots,f(P_{n})|f\in {{\cal L}_{q}}(D,G)\},
\end{equation} 
and
\begin{equation}
C_{{\cal L}_{q}}(D,G;{\bf v}):=\{(v_1f(P_{1}),\hdots,v_nf(P_{n}))|f\in {\cal L}_{q}(D,G)\}.
\end{equation}

We now give sufficient conditions for an algebraic geometry $C_{{\cal L}_{q}}(D,G;{\bf v})$ to be Hermitian self-orthogonal.
\begin{thm}\label{thm:key} Let $\cal X$ be a smooth projective curve with genus $g$. Let $D=P_{1}+\cdots+P_{n}$ and $G=(k-1)P_\infty$ be two divisors, $\omega$ be a Weil differential form such that $H=D-G+(\omega)$ and ${\bf v}=(v_1,\hdots,v_n)$ with $v_i\in \F_{q^2}^*.$ Then the code $C_{{\cal L}_{q}}(D,G;{\bf v})$ is Hermitian self-orthogonal over $\F_{q^2}$ if the following conditions hold
\begin{enumerate}[1)]
\item $v_{P_i}(\omega)=-1$ for $1\le i\le n$,
\item $G\le H$,
\item $\text{Res}_{P_i}(\omega)=v_i^{q+1}$ for $1\le i\le n$,
\item $g+1\le k\le \lfloor \frac{n+q+2g-1}{q+1} \rfloor$.
\end{enumerate}
\end{thm}

\begin{proof} Let ${\bf c}=(v_1f(P_{1}),\hdots,v_nf(P_{n}))\in C_{{\cal L}_{q}}(D,G;{\bf v})$ with $f\in {\cal L}_q(D,(k-1)P_\infty)$. Then we have the following equivalences:
 \begin{equation}
 \begin{array}{ll}
 {\bf c}=(v_1f(P_{1}),\hdots,v_nf(P_{n}))\in C_{{\cal L}_{q}}(D,G;{\bf v})^{\perp_H}\\
 \Longleftrightarrow 
 (v_1^qf^q(P_{1}),\hdots,v_n^qf^q(P_{n}))\in C_{{\cal L}_{q}}(D,G;{\bf v})^{\perp}\\
 \Longleftrightarrow 
 (v_1^qf^q(P_{1}),\hdots,v_n^qf^q(P_{n}))\in \left(\frac{\text{Res}_{P}(\omega)}{\bf v}\right)\cdot C_{{\cal L}_{q}}(D,H)\\
  \Longleftrightarrow 
 (v_1f(P_{1}),\hdots,v_nf(P_{n}))\in \left(\frac{\text{Res}_{P}(\omega)}{\bf v^{q+1}}\right)\cdot C_{{\cal L}_{q}}(D,H;{\bf v}),
 \end{array}
 \label{eq:hermitian-dual}
 \end{equation}
where $\frac{\text{Res}_{P}(\omega)}{\bf v^{q+1}}=(\frac{\text{Res}_{P_1}(\omega)}{v_1^{q+1}},\hdots, \frac{\text{Res}_{P_n}(\omega)}{v_n^{q+1}})$. The first equivalence holds due to (\ref{eq:E-H-dual}), the second one is due to Lemma \ref{lem:1}, and the last one is by definition (\ref{eq:equivalent-code}) and from the fact that $f\in {\cal L}_q(D,G)$.
 Then $f$ has all poles of degree at most $k-1$, and it thus has at most $k-1$ zeros, which means that the degree of $f$ is upper bounded by $k-1$. From the second equivalence, we have that $f^q$ has degree at most $n+2g-k-1$, and thus it is sufficient to take $(k-1)q\le n+2g-k-1,$ that is, $k\le \lfloor \frac{n+q+2g-1}{q+1} \rfloor$. By combining the above equivalences (\ref{eq:hermitian-dual}) with points 1) and 2), we get the result as claimed.
\end{proof}

The following lemma is useful for embedding a Hermitian self-orthogonal code constructed from projective lines.
\begin{lem}\label{lem:embedding} Assume that $G=(k-1)P_\infty$, $1\le k(q+1)\le{n+q-1}$ and $\text{Res}_{P_i}(\omega)=v_i^{q+1}$ for some $\beta \in  \F_{q^2}^*$ for $1\le i\le n$.  Then an MDS Hermitian self-orthogonal code $C_{\cal L}(D,G;{\bf v})$ with parameters $[n,k]_{q^2}$ can be embedded into an MDS Hermitian self-orthogonal $[n',k+1]_{q^2}$ code, where $n'$ is determined by:
$$
n'=
\begin{cases}
n+1\textnormal{ if }k(q+1)=n-1,\\
n\textnormal{ if }n-1+q\not= k(q+1)\not= n-1.
\end{cases}
$$
\end{lem}
\begin{proof} Let $C_{\cal L}(D,G;{\bf v})$ be an MDS Hermitian $q^2$-ary self-orthogonal $[n,k]_{q^2}$ code with its generator matrix
\begin{equation*}
{\cal G}=\left(
\begin{array}{cccc}
v_1&v_2&\hdots&v_{n}\\
v_1\alpha_1&v_2\alpha_2&\cdots&v_{n}\alpha_{n}\\
\vdots&\vdots&\cdots&\vdots\\
v_1\alpha_1^{k-1}&v_2\alpha_2^{k-1}&\cdots&v_{n}\alpha_{n}^{k-1}\\
\end{array}
\right).
\end{equation*}

Consider the code ${\bar C}_{\cal L}(D,G;{\bf v})$ with generator matrix
\begin{equation}\label{eq:gen2}
\bar{\cal G}=\left(
\begin{array}{ccccc}
v_1&v_2&\hdots&v_{n}&0\\
v_1\alpha_1&v_2\alpha_2&\cdots&v_{n}\alpha_{n}&0\\
\vdots&\vdots&\cdots&\vdots&0\\
v_1\alpha_1^{k-1}&v_2\alpha_2^{k-1}&\cdots&v_{n}\alpha_{n}^{k-1}&0\\
v_1\alpha_1^{k}&v_2\alpha_2^{k}&\cdots&v_{n}\alpha_{n}^{k}&\alpha \\
\end{array}
\right)
=\left(
\begin{array}{c}
\bar{g}_1\\
\bar{g}_2\\
\vdots\\
\bar{g}_{k}\\
\bar{g}_{k+1}\\
\end{array}
\right),
\end{equation}
for some $\alpha\in \F_{q^2}$ such that $\alpha^{q+1}=-1$.

Since ${C}_{\cal L}(D,G;{\bf v})$ is Hermitian self-orthogonal, ${\bf v}^{q+1}=(v_1^{q+1},\hdots,v_n^{q+1})$ is a solution to the following system of equations:
\begin{equation}\label{eq:square1}
M_{(q+1)(k-1)}{\bf x}^{{q+1}^\top}=(0,\hdots,0)^\top ,
\end{equation}
where
$
M_i:=\left(
\begin{array}{cccc}
1&1&\cdots&1\\
\alpha_1&\alpha_2&\cdots&\alpha_n\\
\vdots&\vdots&\cdots&\vdots\\
\alpha_1^{i}&\alpha_2^{i}&\cdots&\alpha_n^{i}\\
\end{array}
\right).
$

Now consider the extended system of (\ref{eq:square1}) with $(q+1)$ additional equations as follows:
\begin{equation}\label{eq:square2}
M_{(q+1)k}{\bf x}^{{q+1}^\top}=(0,\hdots,0,\underbrace{0,\hdots,0}\limits_{q},1)^\top.
\end{equation}
It follows that there exists a Hermitian self-orthogonal code with parameters $[n+1,k+1]_{q^2}$ if the system of equations, defined by (\ref{eq:square2}), has a solution.
Note that the largest positive integer $k'$, with $k'(q+1)\le {n+q-1}$ such that $M_{n}$ is the Vandermonde matrix, satisfies $n=k'(q+1)+1$. 

If $k(q+1)\not=n-1$, then the matrix $M_{(q+1)k}$ is not a square matrix, and thus the system (\ref{eq:square2}) can only have a trivial solution for $\alpha=0$. In this case, we obtain an MDS Hermitian self-orthogonal code with parameters $[n,k+1, n-k]_{q^2}$.
%there exists a Hermitian self-orthogonal code $C'$ with parameters $[n+1,k+1,\ge n-k]_{q^2}$. By puncturing the last coordinate (the last column of $\bar{\cal G}$ in (\ref{eq:gen2})) of $C'$, we obtain an MDS Hermitian self-orthogonal code with parameters $[n,k+1, n-k]_{q^2}$.

If $k(q+1)=n-1$, then the matrix $M_{(q+1)k}$ is invertible.
With  $v_i^{q+1}=\frac{1}{\prod\limits_{j=1,j\not= i}^n{(\alpha_i-\alpha_j)}}$ for $1\le i \le n$,  we have that ${\bf v}^{q+1}$   is also the (unique) solution to (\ref{eq:square2}) (by Cramer's rule). It implies that $<{\bar g}_{i},{\bar g}_{k+1}>_H=0$ for $1\le i\le k+1$.
The subcode $\bar{C}_{\cal L}(D,G;{\bf v})^{(k)}$ of $\bar{C}_{\cal L}(D,G;{\bf v})$, generated by the first $k$ rows  $\bar {g}_1,\hdots,\bar{g}_k$, is a Hermitian self-orthogonal code with parameters $[n+1,k]_{q^2}$. It can be easily checked that ${\bar g}_{k+1}$ is not in the subspace spanned by ${\bar g}_{1},\hdots, {\bar g}_{k}$. Hence, the extended code $\bar {C}_{\cal L}(D,G;{\bf v})$ is an MDS Hermitian self-orthogonal code with parameters $[n+1,k+1]_{q^2}.$
\end{proof}

\subsection {Hermitian self-orthogonal codes from projective lines}\label{subsection:projective-lines}
In this subsection, we consider Hermitian self-orthogonal codes over $\F_{q^2}$ from the projective lines ${\mathbb P}^{1}=\{(\alpha:1)|\alpha\in \F_{q^2}\}\cup \{(1:0)\}$. There are $q^2+1$ distinct points in ${\mathbb P}^{1}$, and they give rise to $q^2+1$ distinct places of $\F_{q^2}({\mathbb P}^1)$. 

We now construct families of MDS Hermitian self-orthogonal codes $C_{\cal L}(D,G)$ by fixing the divisor $G$. Our constructions are done by choosing, as sets of evaluation points, some special subsets of $\F_{q^2}$, some multiplicative subgroups of $\F_{q^2}^*$ and their cosets.
The first construction deals with a multiplicative subgroup of $\F_{q^2}^*$ of order $n-1$, and it is given as follows.
\begin{const} \label{thm:MDS1}
\textnormal{
 Let $q=p^m$, $(n-1)|(q^2-1)$ and $1\le k\le \lfloor \frac{n+q-1}{q+1} \rfloor$. 
Take $U=\{\alpha\in \F_{q^2}|\alpha^n=\alpha\}=\{\alpha_1,\hdots,\alpha_n\}$ with $\alpha_n=0$ and $h(x)=x^n-x$. Set $D=(h)_0$, $G=(k-1)P_\infty$, $\omega=\frac{dx}{h(x)}$ and $H=(n-k-1)P_\infty$. The polynomial $h(x)$ has $n$ simple roots and $v_{P_i}(\omega)=-1$ for any $1\le i\le n$, which gives point 1) of Theorem \ref{thm:key}. It should be noted that for the projective lines, we have $g=0$, and thus the condition $k\le \lfloor \frac{n+q-1}{q+1}\rfloor $ implies point 4) of Theorem \ref{thm:key}, and it also implies that $G\le H$, which gives point 2) of Theorem \ref{thm:key}. Moreover, for any $1\le i \le n$, we have $h'(\alpha_i)=-1\text{ or }n-1$ which are elements in $\F_q^*$, and thus $h'(\alpha_i)=\beta_i^{q+1}$ for some $\beta_i\in \F_{q^2}$, that is, $\text{Res}_{P_i}(\omega)=1/(\beta_i)^{q+1}$, which proves point 3) of Theorem \ref{thm:key}. Combining the previous four arguments, it follows that the code $C_{{\cal L}_{q}}(D,G;{\bf v})$, where ${\bf v}=(v_1,\hdots,v_n)$ and $v_i=1/\beta_i$ is Hermitian self-orthogonal over $\F_{q^2}$. 
Hence, there exists an MDS Hermitian self-orthogonal $[n,k]_{q^2}$ code.  
By applying the embedding Lemma \ref{lem:embedding}, it follows that
\begin{itemize}
\item[{\textbf E1.}] there exists an MDS Hermitian self-orthogonal $[n+1,k+1]_{q^2}$ code if $k(q+1)=n-1$;
\item[{\textbf E2.}] there exists an MDS Hermitian self-orthogonal $[n',k+1]_{q^2}$ code if ${n+q-1}\not=k(q+1)\not=n-1$, where $n'=n$ or $n'=n+1$. 
\end{itemize}
Under the special condition $(n-1)|(q^2-1)$, there could exist an MDS Hermitian self-orthogonal code with parameters $[n,s]_{q^2}$ with $s$ strictly greater than $k+1$. 
Now, we assume that such a code exists, and it has the following generator matrix:
\begin{equation*}
{\cal G}=\left(
\begin{array}{cccc}
v_1&v_2&\hdots&v_{n}\\
v_1\alpha_1&v_2\alpha_2&\cdots&v_{n}\alpha_{n}\\
\vdots&\vdots&\cdots&\vdots\\
v_1\alpha_1^{s-1}&v_2\alpha_2^{s-1}&\cdots&v_{n}\alpha_{n}^{s-1}\\
\end{array}
\right).
\end{equation*}
We will explore what values can be taken by $s$.
For $1\le i\le n$, set $w_i=v_i^{\frac{q+1}{2}}$. Let $k'\le \frac{n}{2}$, and put $G'=(k'-1)P_\infty$ as well as $H'=D-G'+(\omega)$. Since $\text{Res}_{P_i}(\omega)={w_i}^2$ for $1\le i\le n$ and $G'\le H'$, the code $C_{{\cal L}}(D,G';{\bf w})$, where ${\bf w}=(w_1,\hdots,w_n)$, is Euclidean self-orthogonal with parameters $[n,k']_{q^2}$. Consider a generator matrix of $C_{{\cal L}}(D,G';{\bf w})$ with the following form:
\begin{equation}
{\cal G}'=\left(
\begin{array}{ccccc}
w_1&w_2&\hdots&w_{n}\\
w_1\alpha_1&w_2\alpha_2&\cdots&w_{n}\alpha_{n}\\
\vdots&\vdots&\cdots&\vdots\\
w_1\alpha_1^{k'-1}&w_2\alpha_2^{k'-1}&\cdots&w_{n}\alpha_{n}^{k'-1}\\
%w_1\alpha_1^{k}&w_2\alpha_2^{k}&\cdots&w_{n}\alpha_{n}^{k} \\
\end{array}
\right)
=\left(
\begin{array}{c}
{g}'_1\\
{g}'_2\\
\vdots\\
{g}'_{k'}\\
%{g}'_{k+1}\\
\end{array}
\right).
\end{equation}
It follows that for $1\le i,j\le k'$,
\begin{equation}\label{eq:E-gigj}
\langle {g}'_i,{g}'_j\rangle_E=\sum\limits_{l=1}^n{w_l^2\alpha_l^{(i-1)+(j-1)}}=0,
\end{equation}
and
\begin{equation}\label{eq:E-gigi}
\langle {g}'_i,{g}'_i\rangle_E=\sum\limits_{l=1}^n{w_l^2\alpha_l^{2(i-1)}}=0,
\end{equation}
Now, let us consider the following Hermitian inner product: for $1\le i,j\le k',$
\begin{equation}\label{eq:H-gigj}
\langle {g}_i,{g}_j\rangle_H=\sum\limits_{l=1}^n{v_l^{q+1}\alpha_l^{(i-1)+q(j-1)}}.
\end{equation}
%It follows that
%\begin{equation}\label{eq:H-gigi-0}
%\langle {g}_t,{g}_t\rangle_H=\sum\limits_{i=1}^n{v_l^{q+1}\alpha_l^{(q+1)(t-1)}}
%\end{equation}
By writing 
\begin{equation}\label{eq:B(t)}
(q+1)(j-1)=A(n-1)+B(j),
\end{equation}
with $A$ being a non-negative integer and $0\le B(j)<n-1$, we obtain that
\begin{equation}\label{eq:H-gigi}
\langle {g}_j,{g}_j\rangle_H=\sum\limits_{l=1}^nv_l^{q+1}\alpha_l^{B(j)}.
\end{equation}
For $(q+1)\le n-1$, we have  that $\langle {g}_{j+1},{g}_{j+1}\rangle_H=\sum\limits_{l=1}^nv_l^{q+1}\alpha_l^{(q+1)}$ is equal to $-v_n^{q+1}$ if $(n-1)|(j-1)(q+1)$ and is equal to zero otherwise. 
We now summarize the possible embeddings as follows:
\begin{enumerate}[1)]
\item {\bf General case $(n-1)|(q^2-1)$:}\\
Assume that there exists an MDS Hermitian self-orthogonal code with parameters $[n,k_0']_{q^2}$ satisfying $2k_0'+q\le n$. 
First note that $\langle {g}_{k_0'},{g}_{k_0'}\rangle_H=\sum\limits_{l=1}^nv_l^{q+1}\alpha_l^{B(k_0')}=0$ implies that $B(k_0')\le 2k_0'-2.$
It follows that 
$$
\begin{array}{ll}
\langle {g}_{k_0'},{g}_{k_0'+1}\rangle_H&=\sum\limits_{l=1}^nv_l^{q+1}\alpha_l^{k_0'-1+qk_0'}\\
&=\sum\limits_{l=1}^nv_l^{q+1}\alpha_l^{(k_0'-1)(q+1)+q}\\
&=\sum\limits_{l=1}^nv_l^{q+1}\alpha_l^{B(k_0')+q}\\
&=0.
\end{array}
$$
The last equality holds due to the fact that the exponent $B(k_0')+q\le n-2$. Hence, for any $1\le i\le k_0'$, $\langle {g}_{i},{g}_{k_0'+1}\rangle_H=0$, and we deduce the following embedding:
\begin{enumerate}[a)]
\item if $(n-1)|k_0'(q+1)$, then there exists an MDS Hermitian self-orthogonal code with parameters $[n+1,k_0'+1]$;
\item if $(n-1){\not|}k_0'(q+1)$, then there exists, an MDS Hermitian self-orthogonal code with parameters $[n,k_0'+1]$.
\end{enumerate}
\item {\bf Case $n=2(q-1)+1$:}\\
For any $1\le i\le k'$, it follows from (\ref{eq:H-gigj}) that $\langle {g}_i,{g}_{2j_0-1}\rangle_H=\sum\limits_{l=1}^n{v_l^{q+1}\alpha_l^{(i-1)+2j_0-2}}$ since $\alpha_l^{n}=\alpha_l$ for any $1\le l\le n$, and thus from (\ref{eq:E-gigj}), we obtain that $\langle {g}_i,{g}_{2j_0-1}\rangle_H=0$ if $2j_0-1\le k'$. Similarly, for any $1\le i\le k'$, we obtain that $\langle {g}_i,{g}_{2j_0}\rangle_H=-\sum\limits_{l=1}^n{v_l^{q+1}\alpha_l^{(i-1)+(2j_0-1)}}=0$ if $2j_0\le k'$. Hence, for any $1\le i,j\le \lfloor \frac{n}{4}\rfloor$, we have that $\langle {g}_i,{g}_j\rangle_H=0$. 
which means that there exists an MDS Hermitian self-orthogonal code, say $C$, with parameters $[n=2(q-1)+1,k''=\lfloor \frac{n}{4}\rfloor]_{q^2}$. 
It follows that the code $C$ can be embedded into an MDS Hermitian self-orthogonal code with parameters $[n',k''+1]_{q^2}$, where 
\begin{equation}\label{eq:n'}
n'=
\begin{cases}
n+1\text{ if }(n-1)|k''(q+1),\\
n\text{ if }(n-1){\not|}k''(q+1).\\
\end{cases}
\end{equation}
\item {\bf Case $n=t(q-1)+1$, with $t|(q+1)$:}\\
If $t|(q+1)$, we consider the Hermitian inner products $\langle {g}_i,{g}_j\rangle_H=0$, with $j=tj_0,tj_0-1,\hdots$. By similar arguments as in point 2), we obtain that for any $1\le i,j\le \lfloor \frac{n}{2t}\rfloor$, $\langle {g}_i,{g}_j\rangle_H=0$. Thus, there exists an MDS Hermitian self-orthogonal code $C$ with parameters $[n,k''=\lfloor \frac{n}{2t}\rfloor]_{q^2}$. It follows that the code $C$ can be embedded into an MDS Hermitian self-orthogonal code with parameters $[n',k''+1]_{q^2}$, where $n'$ is determined by (\ref{eq:n'}).
\end{enumerate}
}
\end{const}
\begin{exam}\label{exam:1}
For $q=13$, $n=25$, and $k=\lfloor \frac{n+q-1}{q+1}\rfloor =2$, we obtain, from Construction \ref{thm:MDS1}, an MDS Hermitian self-orthogonal code $C_2$ with parameters $[25,2,24]_{13^2}$. It can be clearly seen that $(n-1)=24\not=k(q+1)=28\not=(n-1)+q=37$, and thus by applying the embedding Lemma \ref{lem:embedding} (from {\textbf E2.}), we obtain an MDS Hermitian self-orthogonal code $C_3$ with parameters $[25,3,23]_{13^2}$. Since $n-1=2(q-1)$, from point 2) of the above discussion, there exists an MDS Hermitian self-orthogonal code $C_7$ with parameters $[25,7,19]_{13^2}$.
The code $C_7$ has 
the generator matrix $(I_7|A_7|B_7)$, where $A_7$ and $B_7$, respectively, are given as follows:
{\scriptsize
$$
\left(
\begin{array}{llllllllllllllllllllllllll}
\theta^{31 }&\theta^{113 }&\theta^{127 }&\theta^{46 }&\theta^{116 }&\theta^{100 }&\theta^{44 }&\theta^{82 }&\theta^{40 }\\
\theta^{63 }&\theta^{74 }&\theta^{145 }&\theta^{101 }&\theta^{36 }&\theta^{58 }&\theta^{6 }&\theta^{47 }&\theta^{142 }\\
\theta^{109 }&\theta^{155}&{ 11 }&\theta^{74 }&\theta^{13 }&\theta^{30 }&\theta^{12 }&\theta^{54 }&\theta^{15 }\\
\theta^{101 }&\theta^{135 }&\theta^{113 }&\theta^{129 }&\theta^{88 }&\theta^{109 }&\theta^{86 }&\theta^{162 }&\theta^{124 }\\
\theta^{148 }&\theta^{78 }&\theta^{44 }&\theta^{95 }&\theta^{94 }&\theta^{135 }&\theta^{116 }&\theta^{19 }&\theta^{15 }\\
\theta^{52 }&\theta^{76 }&\theta^{106 }&\theta^{145 }&\theta^{11 }&\theta^{92 }&\theta^{93}&{ 1 }&\theta^{159 }\\
\theta^{11 }&\theta^{82 }&\theta^{38 }&\theta^{141 }&\theta^{163 }&\theta^{111 }&\theta^{152 }&\theta^{79 }&\theta^{74 }\\
\end{array}
\right),
$$
$$
\left(
\begin{array}{llllllllllllllllllllllllll}
\theta^{61 }&\theta^{145 }&\theta^{149 }&\theta^{79 }&\theta^{137 }&\theta^{109 }&\theta^{64 }&\theta^{92 }&\theta^{52}\\
\theta^{137 }&\theta^{23 }&\theta^{145 }&\theta^{153 }&\theta^{86 }&\theta^{113 }&\theta^{59 }&\theta^{152 }&\theta^{130}\\
\theta^{5 }&\theta^{63 }&\theta^{37 }&\theta^{85 }&\theta^{53 }&\theta^{68 }&\theta^{78 }&\theta^{97 }&\theta^{122}\\
\theta^{148 }&\theta^{33 }&\theta^{11 }&\theta^{79 }&\theta^{87 }&\theta^{137 }&\theta^{135 }&\theta^{50}&{ \theta}\\
\theta^{40 }&\theta^{127 }&\theta^{100 }&\theta^{4 }&\theta^{32 }&\theta^{122 }&\theta^{155 }&\theta^{58 }&\theta^{73}\\
\theta^{50 }&\theta^{138 }&\theta^{145 }&\theta^{44 }&\theta^{76 }&\theta^{18 }&\theta^{91 }&\theta^{29 }&\theta^{32}\\
\theta^{128 }&\theta^{82 }&\theta^{90 }&\theta^{23 }&\theta^{50 }&\theta^{164 }&\theta^{89 }&\theta^{67 }&\theta^{105}\\
\end{array}
\right),
$$
}
respectively, where  $\theta$ is a primitive element of $\F_{13^2}$. 
%}
\end{exam}

\begin{exam} For $q=11$, $n-1=15$ ($n\not=2(q-1)+1$), and $k=2$, we obtain, from Construction \ref{thm:MDS1}, MDS Hermitian self-orthogonal codes with parameters $[16,2]_{11^2}$ and $[16,3]_{11^2}$ (from {\textbf E2.}), respectively. Now,  we have $k'=3$. Since $15{\not|} (3\times12 )$, from point 1) b), there exists an MDS Hermitian self-orthogonal code with parameters $[16,4]_{11^2}$ (confirmed by Magma). 
According to point 1) b), one may try to check the case $k'=4$ by Magma, but this is not in the range of our discussion above, where $k'$ is required to be less than or equal to $\lfloor \frac{n}{4}\rfloor -1=3$. For $k'=4$, it follows that $(n-1)|(k'+1)(q+1)$. However, it is confirmed, by Magma, that the code with parameters $[16,5]$ is no longer a Hermitian self-orthogonal code.
For $q=11$ and $n=3(q-1)+1=31$, it follows from point 3) of Construction \ref{thm:MDS1} that there exist MDS Hermitian self-orthogonal codes with parameters $[31,5]_{11^2}$. Now $k'=5$ and $k'(q+1)=60\equiv 0 \pmod {n-1}$, and thus there exists an MDS Hermitian self-orthogonal code with parameters $[32,6]_{11^2}$. 
It should be noted that $k'=6$ is the maximal dimension such that $M_{k'(q+1)}$ associated to (\ref{eq:square1}) is a square matrix, and thus we can no longer embed the Hermitian self-orthogonal code with dimension $k'$ into that of dimension $k'+1$. 
\end{exam}

To provide some constructions from cosets of the multiplicative subgroups of $\F_{q^2}^*$, we need the following lemma to fulfill the desired properties of those cosets.
\begin{lem} \label{lem:alpha_ij}Let $q=p^m$ with $p$ an odd prime, $n=\frac{q-1}{p^r+1}$ and for $\alpha_i,\alpha_j\in \F_q$ with $\alpha_i\not =\alpha_j$, denote $\alpha_{ij}=\alpha_i^n-\alpha_j^n.$ Let $\omega$ be a primitive element of $\F_q$. If $r|m$, then $\alpha_{ij} \in \F_{p^r}.$
\end{lem}
\begin{proof}
Raising $\alpha_{ij}$ to the power $p^r,$ we get
$$
\alpha_{ij}^{p^r}={(\alpha_i^{n{p^r}}-\alpha_j^{np^r})}
={(\alpha_i^{q-1+n}-\alpha_j^{q-1+n})}=\alpha_i^n-\alpha_j^n=\alpha_{ij}.
$$
Thus, the result follows.
\end{proof}
We can now give more constructions of MDS Hermitian self-orthogonal codes from the cosets.
\begin{const}\label{thm:multi-cosets}
\textnormal{
 Let $m=2s$, $q = p^m$, $q_0=p^s$ be an odd prime power, $n= \frac{q-1}{p^r+1}$ even, $1\le r<m$, $r|\frac{m}{2}$ and $1\le k\le \lfloor \frac{(t+1)n+q_0-1}{q_0+1} \rfloor$. Set $U_n=\{\alpha\in \F_q| \alpha^n=1\}$, say $U_n=\{u_1,\hdots,u_n\}$. Assume that there exist $t$ distinct multiplicative cosets with their coset leaders $\alpha_1,\hdots,\alpha_t$ are elements in $\F_{q_0^2}$ of power $q_0+1$. 
 Let $\alpha_1U_n,\hdots,\alpha_tU_n$ be $t$ nonzero cosets of $U_n$. Put $U=U_n\cup\{0\}\cup \left(\bigcup\limits_{i=1}^t\alpha_iU_n\right)=\{a_1,\hdots,a_{(t+1)n},a_{(t+1)n+1}\}$, and write
$$h(x)=\prod\limits_{\alpha\in U}(x-\alpha).$$
Then the derivative of $h(x)$ is given by 
$$
\begin{array}{ll}
h'(x)=&((n+1)x^n-1)\prod\limits_{i=1}^t(x^n-\alpha_i^n)\\
&+nx^n(x^n-1)\left(\sum\limits_{i=1}^t\prod\limits_{j=1,j\not=i}^t(x^n-\alpha_j^n)\right).\\
\end{array}
$$
For $1\le j\le t,1\le s\le n$, we have 
\begin{equation}
h'(\alpha_ju_s)=n\alpha_j^n(1-\alpha_j^n)\prod\limits_{i=1,i\not=j}^t(\alpha_j^n-\alpha_i^n).
\label{eq:derivative}
\end{equation}
For $1\le i,j\le t$ and $n=\frac{q-1}{p^r+1},$ we have, from Lemma \ref{lem:alpha_ij},
$\alpha_{ji}=\alpha_j^n-\alpha_i^n\in {\F_{p^r}}$.
 Hence if $r|\frac{m}{2}$, then $\alpha_{ji}=\beta_{ji}^{q+1}$ for some $\beta_{ji}\in \F_{p^m}$
Thus for any $a_i\in U, h'(a_i)=\beta_i^{q+1}$ for some $\beta\in \F_{q^2}$.
Clearly, all the $(t+1)n+1$ roots of $h(x)$ are simple. Set $G=(k-1)P_\infty$, $D=(h)_0$, $\omega=\frac{dx}{h(x)}$ and $H=((t+1)n-k-1)P_\infty$. Thus by Theorem \ref{thm:key}, the code $C_{{\cal L}_{q}}(D,G;{\bf v})$, where ${\bf v}=(v_1,\hdots,v_{(t+1)n+1})$ and $v_i^{q+1}=1/\beta_i$, is Hermitian self-orthogonal over $\F_{q^2}$.
Hence, there exists an MDS Hermitian self-orthogonal $[(t+1)n+1,k]_{q_0^2}$ code. By applying the embedding Lemma \ref{lem:embedding}, it follows that
\begin{itemize}
\item there exists an MDS Hermitian self-orthogonal $[(t+1)n+2,k+1]_{q^2}$ code if $k(q+1)=(t+1)n$;
\item there exists an MDS Hermitian self-orthogonal $[(t+1)n+1,k+1]_{q^2}$ code if ${(t+1)n+q}\not=k(q+1)\not=(t+1)n$.
\end{itemize}
}
\end{const}
\begin{exam}\label{exam:2}
%\textnormal{
For $q=9$ and $n=8,t=2$, the maximal dimension of the code in Construction \ref{thm:MDS1} is $k=2$, and thus we obtain an MDS Hermitian self-orthogonal code $C'_2$ with parameters $[17,2,16]_{9^2}$. It can be clearly seen that $16\not=k(q+1)=20\not= (n-1)+q=25$, and thus by applying the embedding Lemma \ref{lem:embedding}, we obtain an MDS Hermitian self-orthogonal code $C'_3$ with parameters $[17,3,15]_{9^2}$. Denote $C'_2={\bf a}\cdot C_2,C'_3={\bf a}\cdot C_3$, where ${\bf a}=(\theta^{76}, 1, 1, 1, 1, 1, 1, 1, 1, 1, 1, 1, 1, 1, 1, 1, 1)$ with $\theta$ being a primitive element of $\F_{9^2}$. We give the generator matrix $G_3=(A_3|B_3)$ of $C_3$ as follows:
{\scriptsize
$$
A_3=\left(
\begin{array}{lllllllll}
\theta^{76} &1& 1& 1& 1& 1& 1& 1& 1\\
0&2&\theta^{45 }&\theta^{50 }&\theta^{55 }&\theta^{60 }&\theta^{65 }&\theta^{70 }&\theta^{75}\\
0&1&\theta^{10 }&\theta^{20 }&\theta^{30}&{ 2 }&\theta^{50 }&\theta^{60 }&\theta^{70}\\
\end{array}
\right),
$$
$$
B_3=\left(
\begin{array}{llllllllc}
1& 1& 1& 1& 1& 1& 1& 1\\
{ 1 }&\theta^{5 }&\theta^{10 }&\theta^{15 }&\theta^{20 }&\theta^{25 }&\theta^{30 }&\theta^{35}\\
{ 1 }&\theta^{10 }&\theta^{20 }&\theta^{30}&{ 2 }&\theta^{50 }&\theta^{60 }&\theta^{70}\\
\end{array}
\right).
$$
}
Note that the code $C_3$ may not be Hermitian self-orthogonal while the code $C'_3$ is, and the first two rows of $G_3$ generates a code $C_2$ that is equivalent to the Hermitian self-orthogonal code $C'_2$ mentioned above, that is, $C'_2={\bf a}\cdot C_2$.
%}
\end{exam}

We can now provide the last explicit constructions of MDS Hermitian self-orthogonal codes from many cosets. The codes obtained have various lengths, and some of the intermediate lengths are different from those in Construction \ref{thm:MDS1}.

\begin{const}\label{const:multicoset2}
\textnormal{
Assume that $n \mid (q^{2}-1)$. Put 
\begin{equation}
n_{1}=\gcd (n, q+1) \text{ and }n_{2}=\frac{n}{\gcd (n,q+1)}. 
\label{eq:n1n2}
\end{equation}
The first equality of (\ref{eq:n1n2}) implies that $n_2$ and $\frac{q+1}{n_1}$ are coprime, and from the assumption, we get that $n_2|(q-1)\frac{q+1}{n_1}$, and thus $n_2|(q-1)$.
Let $U_n$ and $V_n$ be two subgroups of $\mathbb{F}_{q^{2}}^{*}$ generated by $\theta^{\frac{q^{2}-1}{n}}$ and $\theta^{\frac{q+1}{n_{1}}}$, respectively, where $\theta$ is a primitive element of $\mathbb{F}_{q^{2}}$. Then $|U_n|=n$ and $|V_n|=(q-1)n_{1}$. Now $\frac{q^{2}-1}{n}=\frac{q+1}{n_{1}}\cdot\frac{q-1}{n_{2}}$ implies that $ \frac{q+1}{n_{1}} \mid \frac{q^{2}-1}{n}$, and we deduce that $U_n$ is a subgroup of $V_n$.
Let $\alpha_{1}U_n, \ldots , \alpha_{\frac{q-1}{n_{2}}-1}U_n $ be all the distinct cosets of $V_n$ different from $U_n$.
For $1 \leq t \leq \frac{q-1}{n_{2}}-1$, put $U=U_n \bigcup^{t}\limits_{j=1}\alpha_{j}U_n\cup \{0\}$, say $U=\{a_1,\hdots,a_{(t+1)n+1}\}$, and  write $$h(x)=\prod\limits_{\alpha\in U}(x-\alpha).$$
As in {\bf Construction \ref{thm:multi-cosets}}, we check the values of the derivative $h'$ given by (\ref{eq:derivative}).
Since $\alpha_{j}$ is in $V_n$, we can write $\alpha_{j}=\theta^{e_j\frac{q+1}{n_{1}}}$ for some positive integer $e_j$. Thus
$\alpha_{j}^{n}=\theta^{e_jn\frac{q+1}{n_{1}}}=\theta^{e_jn_{2}(q+1)}$, and it is an element of $\mathbb{F}^{*}_{q}$. 
It implies that for any $\alpha\in U$, we have $h'(\alpha)\in \F_{q}^*$, that is for any $1\le i\le (t+1)n+1$, we have $h'(a_i)=\beta_i^2$ for some $\beta_i \in \F_{q^2}$.
Set $G=(k-1)P_\infty$, $D=(h)_0$, $\omega=\frac{dx}{h(x)}$ and $H=((t+1)n-k-1)P_\infty$. Thus by Theorem \ref{thm:key}, the code $C_{{\cal L}_{q}}(D,G;{\bf v})$, where ${\bf v}=(v_1,\hdots,v_{(t+1)n+1})$ and $v_i^{q+1}=1/\beta_i$, is Hermitian self-orthogonal over $\F_{q^2}$.
Hence, for $n|(q^2-1)$, $n_2=\frac{n}{\gcd (n,q+1)}$, $1\le t\le \frac{q-1}{n_2}-1$ and $1\le k\le \frac{(t+1)n+q}{q+1}$, there exists an MDS Hermitian self-orthogonal $[(t+1)n+1,k]_{q^2}$ code. By applying the embedding Lemma \ref{lem:embedding}, it follows that
\begin{itemize}
\item there exists an MDS Hermitian self-orthogonal $[(t+1)n+2,k+1]_{q^2}$ code if $k(q+1)=(t+1)n$;
\item there exists an MDS Hermitian self-orthogonal $[(t+1)n+1,k+1]_{q^2}$ code if ${(t+1)n+q}\not=k(q+1)\not=(t+1)n$.
\end{itemize}
}
\end{const}

We give an example as follows.

\begin{exam}\label{exam:3}
%\textnormal{
For $q=17$, $n=12$ and $t=2$, we obtain, from Construction \ref{const:multicoset2}, an MDS Hermitian self-orthogonal code with parameters $[25,2,24]_{17^2}$. It can be easily checked that $n-1=24\not=k(q+1)=36\not=31$, and thus from the embedding Lemma \ref{lem:embedding}, there exists an MDS Hermitian self-orthogonal code with parameters $[26,3,24]_{17^2}$. 
%}
\end{exam}

\begin{const}\label{thm:MDS-new1}
\textnormal
{ Let $q$ be a prime power. 
Fix an element $\alpha\in \F_{q^2}\backslash \F_q$. Write $\F_q$ as $\{u_1,\hdots, u_q\}$, and denote $\alpha_{i,j}=u_i\alpha+u_j$ for $1\le i\le t$ and $1\le j\le q$. Put $U=\{\alpha_{i,j}|1\le i\le t, 1\le j\le q\}$,
and 
\begin{equation*}
h(x)=\prod\limits_{\beta\in U}(x-\beta)=\prod\limits_{
1\le i\le t,1\le j\le q
}(x-\alpha_{i,j}).
\end{equation*}
Then the derivative $h'(x)$ at $\alpha_{i_0,j_0}\in U$ is given as follows:
\begin{equation*}
\begin{array}{ll}
h'(\alpha_{i_0,j_0})
&=
\prod\limits_{
\begin{array}{c}
1\le i\le t,1\le j\le q\\
(i,j)\not= (i_0,j_0)\\
\end{array}
}(\alpha_{i_0,j_0}-\alpha_{i,j})\\
&=
\prod\limits_{
\begin{array}{c}
1\le j\le q\\
j\not= j_0\\
\end{array}
}(u_{i_0}\alpha+u_{j_0}-u_{i_0}\alpha-u_j)\\
&\prod\limits_{
\begin{array}{c}
1\le i\le t,\\
1\le j\le q\\
i\not= i_0\\
\end{array}
}(u_{i_0}\alpha+u_{j_0}-u_{i}\alpha-u_j)\\
&=
\prod\limits_{
\begin{array}{c}
1\le j\le q\\
j\not= j_0\\
\end{array}
}(u_{j_0}-u_j)\\
&\prod\limits_{
\begin{array}{c}
1\le i\le t\\
1\le j\le q\\
i\not= i_0\\
\end{array}
}((u_{i_0}-u_i)\alpha+(u_{j_0}-u_j))\\
&=
-
\prod\limits_{
\begin{array}{c}
1\le i\le t\\
i\not= i_0\\
\end{array}
}((u_{i_0}-u_i)^q\alpha^q+(u_{j_0}-u_j)\alpha)\\
&=
-(\alpha^q-\alpha)^{t-1}
\prod\limits_{
\begin{array}{c}
1\le i\le t\\
i\not= i_0\\
\end{array}
}(u_{i_0}-u_i).\\
\end{array}
\end{equation*}
The two last equalities hold due to the fact that the product of all element in $\F_q^*$ is equal to $-1$.
This shows that for any $\alpha_{i,j}\in U$, we have that $(\alpha^q-\alpha)^{1-t}h'(\alpha_{i,j})$ is an element of $\mathbb{F}^{*}_{q}$, and thus $(\alpha^q-\alpha)^{1-t}h'(\alpha_{i,j}) =\beta_{i,j}^{q+1}$ for some $\beta_{i,j} \in \F_{q^2}$.
Set $\omega=(\alpha^q-\alpha)^{t-1}\frac{dx}{h(x)}$, $D=(h)_0$, $G=(k-1)P_\infty$ and $H=(n-1-k)P_\infty$. Thus from Theorem \ref{thm:key}, the code ${\bf v}\cdot C_{{\cal L}_{q}}(D,G)$, where ${\bf v}=(v_{i,j})_{1\le i\le t,1\le j\le q}$ and $v_{i,j}=1/\beta_{i,j}$, is an MDS Hermitian self-orthogonal $[tq,k]_{q^2}$ code. By applying the embedding Lemma \ref{lem:embedding}, it follows that
\begin{itemize}
\item there exists an MDS Hermitian self-orthogonal $[tq+1,k+1]_{q^2}$ code if $k(q+1)=tq-1$;
\item there exists an MDS Hermitian self-orthogonal $[n,k+1]_{q^2}$ code, where $n=tq$ or $n=tq+1$, if ${(t+1)q-1}\not=k(q+1)\not=tq-1$.
\end{itemize}
}
\end{const}

\begin{exam}\label{exam:4}
%\textnormal{
 For $q=7$ and $n=21$, the maximal dimension of the code in Construction \ref{thm:MDS-new1} is $k=3$, and thus we obtain an MDS Hermitian self-orthogonal code $D'_3$ with parameters $[21,3,19]_{7^2}$. It can be clearly seen that $(n-1)=20\not=k(q+1)=24\not= (n-1)+q=27$, and thus by applying the embedding Lemma \ref{lem:embedding}, we obtain an MDS Hermitian self-orthogonal code $D'_4$ with parameters $[21,4,18]_{7^2}$. Denote $D'_4={\bf a}\cdot D_4$, where ${\bf a}=( 1 ,1,\theta^{47}, 1 ,1 ,1, 1 ,1 ,1,\theta^{47}, 1,\theta^{47},\theta^{47}, 1, 1 ,1,\theta^{47},\theta^{47}, 1,\theta^{47}, 1 )$ with $\theta$ being a primitive element of $\F_{7^2}$. We give  the generator matrix $G'_4=(A'_4|B'_4)$ of $D_4$ as follows:
{\scriptsize
$$
A'_4=\left(
\begin{array}{lllllllllll}
1 &1&\theta^{47}& 1 &1 &1& 1 &1 &1&\theta^{47}& 1\\
0&\theta^{25}&\theta^{27}&\theta^{33}&\theta^{34}&\theta^{35}&\theta^{38}&\theta^{41}&\theta^{44}&\theta^{47}& \theta\\
0&\theta^{2}&\theta^{7}&\theta^{18}&\theta^{20}&\theta^{22}&\theta^{28}&\theta^{34}& 5&\theta^{47}&\theta^{2}\\
0&\theta^{27}&\theta^{35}&\theta^{3}&\theta^{6}&\theta^{9}&\theta^{18}&\theta^{27}&\theta^{36}&\theta^{47}&\theta^{3}\\
\end{array}
\right),
$$
$$
B'_4=\left(
\begin{array}{lllllllllll}
\theta^{47}&\theta^{47}& 1& 1 &1&\theta^{47}&\theta^{47}& 1&\theta^{47}& 1\\
\theta^{4}&\theta^{6}&\theta^{9}& 2&\theta^{17}&\theta^{17}&\theta^{18}&\theta^{21}&\theta^{21}&\theta^{23}\\
\theta^{9}&\theta^{13}&\theta^{18}& 4&\theta^{34}&\theta^{35}&\theta^{37}&\theta^{42}&\theta^{43}&\theta^{46}\\
\theta^{14}&\theta^{20}&\theta^{27}& 1&\theta^{3}&\theta^{5}& 3&\theta^{15}&\theta^{17}&\theta^{21} \\
\end{array}
\right).
$$
}
%}
\end{exam}

\subsection{Hermitian self-orthogonal codes from maximal curves}
Curves having many rational points produce long codes. The number of $\mathbb F_q$-rational points of a smooth projective curve $\cal X$ defined over $\mathbb F_q$ is bounded by the well known Hasse-Weil bound:
 $$|\sharp{\cal X}(\mathbb F_q)-(q+1)|\le 2 g \sqrt{q},$$
where $g$ is the genus of ${\cal X}$. Curves attaining the bound are called maximal, and they are of great interest in coding theory. In this subsection, we are considering some maximal curves and employ them to construct Hermitian self-orthogonal AG codes.

For $m$ being even, consider the following affine elliptic curve over $\F_{2^m}$:
\begin{equation}\label{eq:elliptic}
{\cal E}:~y^2+y=x^3+c,
\end{equation}
where
\begin{equation} \label{eq:c}
\begin{cases}
c=0\text{ if }m\equiv 2\pmod 4,\\
\text{Tr}(c):=\sum\limits_{i=0}^{m-1}c^{2^i}=1\text{ if }m\equiv 0\pmod 4.
\end{cases}
\end{equation}
Denote 
\begin{equation}\label{eq:U_c}
U_c=\{\alpha\in \F_{2^m}|\text{Tr}(\alpha^3+c)=0\},
\end{equation} 
where $c$ is defined by (\ref{eq:c}).
Any element $\alpha\in U_c$ gives rise to two rational points on $\cal E$, and we denote these rational points by $P_{\alpha}^{(1)}, P_{\alpha}^{(2)}$. Hence, there are $2\sharp U_c$ finite points and one infinite point $P_\infty.$ It is well known from \cite{Mene} that the number of rational points on $\cal E$ is $2^m+2\sqrt{2^m}+1$, and thus $\sharp U_c=2^{m-1}+\sqrt{2^m}$, and $\cal E$ is a maximal curve ($g=1$).
To be able to explicitly construct Hermitian self-orthogonal codes from the elliptic curve (\ref{eq:elliptic}), we propose the following assumption.\\

\noindent
{\bf Assumption 1:} Assume that for any $\alpha\in U_c$, we have that $h'(\alpha)=\beta^{q+1}$ for some $\beta\in \F_{q^2},$ where $h(x)=\prod\limits_{\alpha\in U_c}(x-\alpha)$, and $U_c$ is defined as in (\ref{eq:U_c}).\\

It should be noted from our Magma check, the above assumption is always true for finitely many values of $q$. Using {\bf Assumption 1}, we can now construct Hermitian self-orthogonal codes from the elliptic curve (\ref{eq:elliptic}).
\begin{const}\label{thm:elliptic}
\textnormal
{
Let $s$ be a positive integer, $m=2s$ and $q=2^s$. Put $h(x)=\prod\limits_{\alpha\in U_c}(x-\alpha)$, $n=\sharp U_c$, $D=(h)_0$, $\omega=\frac{dx}{h(x)}$. Label the elements of $U_c$ by $\alpha_1,\hdots,\alpha_n$. Assume that for any $1\le i \le n$, $h'(\alpha_i)=\beta_i^{q+1}$ for some $\beta_i\in \F_{q^2}$. Set $G=(k-1)P_\infty$ and 
$H=(2n-k+1)P_\infty$, with $2\le k\le \lfloor \frac{2n+q+1}{q+1}\rfloor $. We have that points 1), 3), and 4) of Theorem \ref{thm:key} follow immediately. From the range of $k$, it follows that $G\le H$, and this proves point 2) of Theorem \ref{thm:key}. Thus, the code $C_{{\cal L}_{q}}(D,G;{\bf v})$, where ${\bf v}=(v_1,v_1\hdots,v_{n},v_{n})$ and $v_i=1/\beta_i$ is Hermitian self-orthogonal over $\F_{q^2}$. 
Hence, there exists a Hermitian self-orthogonal $[2n,k-1,\ge 2n-k+1]_{q^2}$ code whose Hermitian dual has parameters $[2n,2n+1-k,\ge k-1]_{q^2}$, with $k\le n+1$.
}
\end{const}
\begin{exam}\label{exam:elliptic}
For $q=4$ and $\theta$ a primitive element of $\F_{4^2}$, by using Magma, we obtain, from Construction \ref{thm:elliptic} for $c=\theta^3$, a Hermitian self-orthogonal code with parameters $[24,4,20]_{4^2}$. We give the generator matrix $(I_4|A''_4|B''_4)$ of such a code as follows:
{\scriptsize
$$
A''_4=\left(
\begin{array}{llllllllllllllllllllllll}
\theta^{13}&\theta^{14}&\theta^{12}&\theta^5&\theta^3&\theta^{14}&\theta^{10}&\theta^9&\theta^{12}&\theta^{11}\\
\theta^{13}&\theta^{14}&\theta^6&\theta^8&\theta^7&\theta^9&\theta^5&\theta^7&\theta^3&\theta^{14}\\
1&0&\theta^3&\theta^3&\theta^7&\theta^7&\theta^4&\theta^4&\theta^4&\theta^4\\
0&1&\theta^4&\theta^4&\theta^3&\theta^3&\theta^4&\theta^4&\theta^2&\theta^2\\
\end{array}
\right),
$$
$$
B''_4\left(
\begin{array}{llllllllllllllllllllllll}
\theta^7&\theta^5&\theta^5&\theta^{12}&1&0&\theta^{14}&\theta^2&\theta^4&\theta^9\\
\theta&\theta^{12}&\theta^2&\theta^{13}&\theta^2&\theta^8&1&\theta^6&\theta&\theta^7\\
\theta^{10}&\theta^{10}&\theta^2&\theta^2&\theta^3&\theta^3&\theta^8&\theta^8&\theta^6&\theta^6\\
\theta&\theta&\theta^{10}&\theta^{10}&\theta^6&\theta^6&1&1&\theta^8&\theta^8\\
\end{array}
\right).
$$
}
Also for $q=2^3,2^4,2^5$, respectively, we obtain, from Magma, Hermitian self-orthogonal codes with parameters $[80,8,72]_{2^6}$, $[288,30,258]_{2^8}$, $[1088,102,986]_{2^{10}}$, respectively.
\end{exam}

Next, we consider the affine hyper-elliptic curve over $\mathbb F_{q^2}$ with $q=2^m$, which is defined by
\begin{equation}
{\cal C}: ~y^2+y=x^{q+1}.
\label{eq:hyper-elliptic}
\end{equation}
For any $\alpha\in \mathbb F_{q^2}$, there exactly exist two rational points $P_{\alpha}^{(1)}, P_{\alpha}^{(2)}$ with $x$-component $\alpha$. 
The set ${\cal C}(\mathbb F_{q^2})$ of all rational points of ${\cal C}$ equal $\{P_{\alpha}^{(1)}| \alpha \in \mathbb F_{q^2}\}\cup \{P_{\alpha}^{(2)}|\alpha \in \mathbb F_{q^2}\} \cup \{ P_\infty\}$. 
Thus the curve $\cal C$ has $1+2q^2$ points, and it is a maximal curve since its genus \cite{Stich73} $g=\frac{q}{2}$ and $1+2q^2=1+q^2+2\frac{q}{2}q$.

\begin{const}\label{thm:hyper-elliptic2}
\textnormal{
Let $q=2^m, m\ge 2$, $(n-1)|(q^2-1)$ and $1+q/2\le k\le \lfloor \frac{2n+2q-1}{q+1} \rfloor$. 
 Consider the affine hyper-elliptic curve defined by (\ref{eq:hyper-elliptic}). Take $U_n=\{\alpha \in \F_{q^2}|\alpha^n=\alpha \}\subseteq {\cal C}(\mathbb F_{q^2})$. Set $G=(k-1)P_\infty$. Put $h(x)=x^n+x$, $D=(h)_0$, $\omega=\frac{dx}{h(x)}$ and $H=(2n+q-k-1)P_\infty$. For any $\alpha\in U_n$, we have $h'(\alpha)=1$. Thus by Theorem \ref{thm:key}, the code $C_{{\cal L}_{q}}(D,G;{\bf v})$, where ${\bf v}=(v_1,v_1\hdots,v_{n},v_{n})$ and $v_i=1$ is Hermitian self-orthogonal over $\F_{q^2}$. 
Hence, there exists a Hermitian self-orthogonal $[2n,k-q/2,\ge 2n-k+1]_{q^2}$ code whose Hermitian dual has parameters $[2n,2n+q/2-k,\ge k+1-q]_{q^2}$.
}
\end{const}

The affine Hermitian curve over $\mathbb F_{q^2}$ with $q=p^m$ is defined by
\begin{equation}
{\cal H}: ~y^q+y=x^{q+1}.
\label{eq:hermitian}
\end{equation}
For any $\alpha\in \mathbb F_{q^2}$, there exactly exist $q$ rational points $P_{\alpha}^{(1)}, P_{\alpha}^{(2)},\hdots, P_{\alpha}^{(q)}$ with $x$-component $\alpha$.
The set ${\cal H}(\mathbb F_{q^2})$ of all rational points of ${\cal H}$ equal $\{P_{\alpha}^{(1)}|\alpha \in \mathbb F_{q^2}\}\cup \{P_{\alpha}^{(2)}| \alpha \in \mathbb F_{q^2}\} \cup \cdots \cup \{P_{\alpha}^{(q)}|\alpha \in \mathbb F_{q^2}\}\cup \{P_\infty\}$. Thus the curve $\cal H$ has $1+q^3$ points, and it is a maximal curve since its genus $g=\frac{q(q-1)}{2}$ and $1+q^3=1+q^2+2\frac{q(q-1)}{2}q$. See \cite[Lemma 6.4.4]{Stich} for the detail.

\begin{const} \label{thm:hermitian-2}
\textnormal{
 Let $q=p^m\ge 4$ and $1+q(q-1)/2\le k\le \lfloor \frac{Nq+q^2-1}{q+1} \rfloor$. 
Consider a set $U$ of size $N$, say $U=\{\alpha_1,\hdots,\alpha_N\}$, with the following four cases:
\begin{enumerate}[i)]
\item $U=\{\alpha\in \F_{q^2}|\alpha^N=\alpha\},$
%\item $U=U_n\cup \beta U_n\cup \{0\}$, where $U_n=\{\alpha\in \F_{q^2}|\alpha^n=1\}$ and $\beta=\theta^{q+1}$ with $\theta$ being a primitive element of $\F_{q^2}$,
\item $U=U_n \bigcup^{t}\limits_{j=1}\alpha_{j}U_n\cup \{0\}$, where $U_n$ and $\alpha_j$ are determined as in {\bf Construction \ref{const:multicoset2}},
\item $U=\{\alpha_{i,j}|1\le i\le t,1\le j\le q\}$, where $\alpha_{i,j}$ is determined as in {\bf Construction \ref{thm:MDS-new1}}. 
\end{enumerate}
Set $h(x)=\prod\limits_{\alpha \in U}(x-\alpha)$, $D=(h)_0$, $\omega=\frac{dx}{h(x)}$, $G=(k-1)P_\infty$ and $H=(Nq+q(q-1)-k-1)P_\infty$. 
From the previous discussions in Subsection \ref{subsection:projective-lines}, we know that for any $1\le i\le N$, $h'(\alpha_i)=\beta_i^{q+1}$ for some $\beta_i\in \F_{q^2}$.
Thus by Theorem \ref{thm:key}, the code $C_{{\cal L}_{q}}(D,G;{\bf v})$, where ${\bf v}=(\underbrace{v_1,\hdots,v_1}_{q},\hdots,\underbrace{v_{N},\hdots,v_{N}}_q)$ and $v_i=1/\beta_i$, is Hermitian self-orthogonal over $\F_{q^2}$. 
Hence, there exists a Hermitian self-orthogonal $[Nq,k-\frac{q(q-1)}{2},\ge Nq-k+1 ]_{q^2}$ code whose Hermitian dual has parameters $[Nq,Nq+q(q-1)/2-k,\ge k+1-q(q-1)]_{q^2}$
if one of the following condition holds
\begin{enumerate}[i)]
\item $(N-1)|(q^2-1)$, 
%\item $N=2n+1$, $m=2s$, $q_0=p^s$ odd, $n= \frac{q-1}{p^r+1}$ even, $n<q_0+1$, $r|\frac{m}{2}$,
\item $N=(t+1)n+1$, $n|(q^2-1)$, $n_2=\frac{n}{\gcd (n,q+1)}$, $1\le t\le \frac{q-1}{n_2}-1$,
\item $N=tq$, $1\le t\le q$.
\end{enumerate} 
}
\end{const}

\begin{const}\label{thm:hermitian-3}
\textnormal{
 Let $q=p^m\ge 4$, $s=q\frac{q^2+1}{2}$ and $1+\frac{(q-1)^2}{4}\le k\le \lfloor \frac{s+\frac{(q-1)^2}{2}+q-1}{q+1} \rfloor$. 
Consider an affine algebraic curve over $\F_{q^2}$ defined by
$${\cal X}: y^{q}+y=x^{\frac{{q}+1}{2}}.$$
The curve has genus $g=\frac{(q-1)^2}{4}$.
Put $$U=\{\alpha\in \F_{q^2}|\exists \beta\in \F_{q^2}\text{ such that }\beta^{q}+\beta=\alpha^{\frac{q+1}{2}}\}.$$ 
The set $U$ is the set of $x$-component solutions to the Hermitian curve whose elements are squares in $\F_{q^2}.$ There are $\frac{q^2+1}{2}$ square elements in $\F_{q^2}$, and they give rise to ${q}\frac{q^2+1}{2}$ rational places. Write $U=\{\alpha_1,\hdots,\alpha_n\}$ with $n=\frac{q^2+1}{2}$. Set $h(x)=\prod\limits_{\alpha\in U}(x-\alpha)$, $D=(h)_0$, $\omega=\frac{dx}{h(x)}$, $G=(k-1)P_\infty$ and $H=(s+\frac{(q-1)^2}{2}-k-1)P_\infty$. 
Then $h(x)=x^n-x$ and $h'(x)=nx^{n-1}-1$. We have that $h'(0)=-1\in \F_q$ and $h'(\alpha)=n-1\in \F_q$ for any $\alpha\in U \backslash \{0\}$. Thus for any $1\le i \le n,$ we have $h'(\alpha_i)=\beta_i^{q+1}$ for some $\beta_i\in \F_{q^2}.$
%Put $D=\sum\limits_{\alpha\in U}\left(P_\alpha^{(1)}+\cdots+P_\alpha^{(q)}\right).$ 
Thus by Theorem \ref{thm:key}, the code $C_{{\cal L}_{q}}(D,G;{\bf v})$, where ${\bf v}=(\underbrace{v_1,\hdots,v_1}_{q},\hdots,\underbrace{v_{n},\hdots,v_{n}}_q)$ and $v_i=1/\beta_i$, is Hermitian self-orthogonal over $\F_{q^2}$.
Hence, there exists a Hermitian self-orthogonal $[s,k-\frac{(q-1)^2}{4},\ge s-k+1 ]_{q^2}$ code whose Hermitian dual has parameters $[s,s+\frac{(q-1)^2}{4}-k,\ge k+1-\frac{(q-1)^2}{2}]_{q^2}$.
}
\end{const}
\subsection{Hermitian self-orthogonal codes from other curves}
In this subsection, we consider codes defined by the affine algebraic curve over $\F_{q^2}$ with $q=p^m$ as follows:
\begin{equation}\label{eq:other-curves}
{\cal Z}:~ y^q-y=x^t.
\end{equation}
 The function fields of $\cal Z$ are Artin-Schreier extensions of $\F_q(x)$. It is well known in \cite{Stich73} that the curve $\cal Z$ has genus $g=\frac{(q-1)(t-1)}{2}$. It should be noted that if $\gcd (q^2-1,t)=1$, then the map $\alpha\mapsto \alpha^t$ is a permutation over $\F_{q^2},$ and in this case, it follows that $\sharp \{\alpha\in \F_{q^2}|\text{Tr}_{\F_{q^2}/\F_q}(\alpha^t)=i\}=\sharp\{\alpha\in \F_{q^2}|\text{Tr}_{\F_{q^2}/\F_q}(\alpha)=i\}=q$ for any $i\in \F_q$, where $\text{Tr}_{\F_{q^2}/\F_q}(\alpha):=\alpha+\alpha^{q}$. In this subsection, we will only consider the case where $\gcd (q^2-1,t)\not=1$, otherwise, we will obtain the Hermitian self-orthogonal codes whose lengths are $q^2$ and which may have worse parameters than those obtained from the projective lines.
\subsubsection{\bf The case $q$ being odd}
\begin{const}\label{thm:other-curves}
\textnormal{
 Let $q=p^m$ be an odd prime power. Let $U=\{\alpha\in \F_{q^2}|\exists \beta\in \F_{q^2}\text{ such that }\beta^q-\beta=\alpha\}$. Assume that $\gcd(t,q+1)|\frac{q+1}{2}.$ Since $\textnormal{Tr}(\beta^q-\beta)=0$ for any $\beta\in\F_{q^2}$, it follows that (\ref{eq:other-curves}) has a solution pair $(\alpha,\beta)\in \F_{q^2}\times \F_{q^2}$ when $\alpha$ is a solution of the following equation:
 \begin{equation}\label{eq:trace=0}
\alpha^t+\alpha^{tq}=0.
\end{equation}
%equivalently,
%$$\alpha=0\textnormal{ or }\alpha^{t(q-1)}=-1.$$
Moreover, any $x$-component $\alpha\in U$, which is a solution of (\ref{eq:trace=0}), gives rise to $q$ different rational places. We now count the roots of (\ref{eq:trace=0}). If $\alpha\not =0$, then, since $\gcd(t,(q+1))|\frac{q+1}{2}$, the equation $$\alpha^{t(q-1)}=-1$$ has $s$ roots satisfying $s=\gcd(t(q-1),(q+1)(q-1)/2)=\gcd(t(q-1),(q+1)(q-1))$.
Label the elements of $U$ by $\alpha_0=0,\alpha_1,\hdots,\alpha_s$, and put $n=s+1$. Set $h(x)=\prod\limits_{\alpha\in U}(x-\alpha)$, $D=(h)_0$, $\omega=\frac{dx}{h(x)}$, $G=(k-1)P_\infty$ and $H=(qn+(q-1)(t-1)-k-1)P_\infty$, with $1+(t-1)(q-1)/2\le k\le \lfloor \frac{(s+1)q+t(q-1)}{q+1} \rfloor$. 
Then $h(x)=x^n-x$ and $h'(x)=nx^{n-1}-1$. We have that $h'(0)=-1\in \F_q$ and $h'(\alpha)=n-1\in \F_q$ for any $\alpha\in U \backslash \{0\}$. Thus for any $1\le i \le n,$ we have $h'(\alpha_i)=\beta_i^{q+1}$ for some $\beta_i\in \F_{q^2}.$
Put $D=\sum\limits_{\alpha\in U}\left(P_\alpha^{(1)}+\cdots+P_\alpha^{(q)}\right).$ 
Thus by Theorem \ref{thm:key}, the code $C_{{\cal L}_{q}}(D,G;{\bf v})$, where ${\bf v}=(\underbrace{v_1,\hdots,v_1}_{q},\hdots,\underbrace{v_{n},\hdots,v_{n}}_q)$ and $v_i=1/\beta_i$, is Hermitian self-orthogonal over $\F_{q^2}$.
Hence, there exists a Hermitian self-orthogonal 
$[nq,k-\frac{(q-1)(t-1)}{2},\ge nq-k+1 ]_{q^2}$ code whose Hermitian dual has parameters 
$[nq,nq+\frac{(q-1)(t-1)}{2}-k,\ge k+1-(q-1)(t-1)]_{q^2}$. Particularly, by taking $t=2,3,4$, respectively, we obtain Hermitian self-orthogonal codes and their dual codes with parameters:
\begin{enumerate}
\item $[q(2(q-1)+1),k-\frac{(q-1)}{2},\ge q(2(q-1)+1)-k+1 ]_{q^2}$ and $[q(2(q-1)+1),q(2(q-1)+1)+\frac{(q-1)}{2}-k,\ge k+1-(q-1)]_{q^2}$ if $q+1\equiv 0\pmod 4$;
\item  $[q(3(q-1)+1),k-(q-1),\ge q(3(q-1)+1)-k+1 ]$ and $[q(3(q-1)+1),q(3(q-1)+1)+(q-1)-k,\ge k+1-2(q-1)]_{q^2}$ if $q+1\equiv 0\pmod 6$;
\item $[q(4(q-1)+1),k-\frac{3(q-1)}{2},\ge q(3(q-1)+1)-k+1 ]_{q^2}$ and $[q(3(q-1)+1),q(3(q-1)+1)+\frac{3(q-1)}{2}-k,\ge k+1-3(q-1)]_{q^2}$ if if $q+1\equiv 0\pmod 8$, respectively.
\end{enumerate}
}
\end{const}
\begin{exam}\label{exam:hermitian-1}
By using Magma, we give some parameters of Hermitian self-orthogonal codes over $\F_{q^2}$ for some values of $q$ and for different values of $t$ as follows:
\begin{enumerate} 
\item for $t=2$,
\begin{itemize} 
 \item $q=3$, we obtain a Hermitian self-orthogonal code with parameters $[15, 3, 12] _{3^2}$, and its dual has parameters $[15, 12, 3] _{3^2}$;
\item $q=7$, we obtain a Hermitian self-orthogonal code with parameters $[91, 5, 84]_{7^2}$, and its dual has parameters $[91, 86, 4]_{7^2}$;
\end{itemize}
\item for $t=3$, $q=5$, we obtain a Hermitian self-orthogonal code with parameters $[65, 3, 60] _{5^2}$, and its dual has parameters $[65, 62, 3] _{3^2}$;
\item for $t=4$,
$q=7$, we obtain a Hermitian self-orthogonal code with parameters $[175, 3, 168] _{7^2}$, and its dual has parameters $[175, 172, 3] _{7^2}$;
\item for $t=5$,
$q=9$, we obtain a Hermitian self-orthogonal code with parameters $[369, 4, 359] _{9^2}$, and its dual has parameters $[369, 365, 3] _{9^2}$.
\end{enumerate}
\end{exam}
 \subsubsection{\bf The case $q$ being even} 
 \begin{const}\label{thm:other-curves-even}
 \textnormal
 { Let $q=2^m$ and $t$ be an odd integer. Put $n=\gcd(t(q-1),(q+1)(q-1))+1$. Let $U=\{\alpha\in \F_{q^2}|\exists \beta\in \F_{q^2}\text{ such that }\beta^q-\beta=\alpha\}$, say $U=\{\alpha_0,\alpha_1,\hdots,\alpha_{n-1}\}$. Set $D=(h)_0$, $G=(k-1)P_\infty$ and $H=(qn+(q-1)(t-1)-k-1)P_\infty$, with $1+(t-1)(q-1)/2\le k\le \lfloor \frac{(s+1)q+t(q-1)}{q+1} \rfloor$. Following the same process as that in Construction \ref{thm:other-curves}, we obtain a Hermitian self-orthogonal code $C_{{\cal L}_{q}}(D,G)$ with parameters $[nq,k-\frac{(q-1)(t-1)}{2},\ge nq-k+1 ]_{q^2}$ , and its Hermitian dual has parameters $[nq,nq+(t-1)(q-1)/2-k,\ge k+1-(q-1)(t-1)]_{q^2}$.
 }
 \end{const}
\begin{exam} \label{exam:hermitian-2}
By using Magma, we give some parameters of Hermitian self-orthogonal codes over $\F_{q^2}$ for some values of $q$ and for different values of $t$ as follows:
\begin{enumerate} 
\item for $t=3$,
$q=8$, we obtain a Hermitian self-orthogonal code with parameters $[176, 4, 168]_{8^2}$, and its dual has parameters $[176, 172, 3]_{2^6}$;
\item for $t=5$,
$q=4$, we obtain a Hermitian self-orthogonal code with parameters $[64, 3, 59]_{4^2}$, and its dual has parameters $[64, 61, 3]_{2^4}$;
\end{enumerate}
\end{exam}

\section{Application to quantum codes}\label{section:application}
In this section, we construct quantum error-correcting codes.
Quantum stabilizer codes are analogues of classical additive codes and they can be constructed from classical linear codes with some properties of Euclidean, Hermitian and symplectic self-orthogonality (see \cite{Cal Rai,Ash Kni,Ket Kla}).

Similar to the classical code, for any $[[n,k,d]]_q$ quantum code, the quantum singleton bound is given by $n\geq k+2d-2$. A quantum code $Q$ is called MDS if it achieves the quantum singleton bound. In order to use our results to construct quantum codes, we need to introduce the following lemma for connection.

\begin{lem}\label{lem:H-construction}\textnormal{(\cite{Ash Kni,Ket Kla})} There exists an $[[n,n-2k, d^{\perp_H}]]_q$ quantum  code whenever there exists a classical Hermitian self-orthogonal  $[n,k]_{q^2}$ code whose dual has minimum distance $d^{\perp_H}$.
\end{lem}

By applying Lemma \ref{lem:H-construction} to the MDS Hemitian self-orthogonal codes obtained from {\bf Constructions \ref{thm:MDS1}-\ref{thm:MDS-new1}}, we obtain the following theorems.
\begin{thm}\label{Q:MDS-new} Let $q=p^m$, $n=t(q-1)+1$ with $t|(q+1)$, and $k=\lfloor \frac{n}{2t}\rfloor$. Then
\begin{enumerate}
\item there exists an MDS quantum code with parameters $[[n+1,n-2k-1,k+2]]_q$ if $(n-1)|k(q+1)$;
\item there exists an MDS quantum code with parameters $[[n,n-2k-1,k+2]]_q$ if $(n-1){\not|}k(q+1)$.
\end{enumerate}
\end{thm}
\begin{rem} First recall that the minimum distance of the quantum MDS codes \cite{FangFu} is $d\le \lfloor \frac{n+q-1}{q+1}\rfloor +2\le 4$ if $n=2q-1$. Quantum MDS codes in Theorem \ref{Q:MDS-new} have larger minimum distances ($\lfloor \frac{n}{4}\rfloor+2$) than the known codes in the literature (\cite{FangFu,JinXing14}), for instance, for $q=13$ and $n=25$, we obtain an MDS quantum code with parameters $[[25,11,8]]_{13}$, and the minimum distance of this code is even larger that obtained from the implicit construction \cite{GraRot} which is only $7$. For $q\ge 13$, our explicitly constructed codes have better minimum distances than those \cite{GraRot} (see Table \ref{table:mds1})
\end{rem}
\begin{thm} \label{Q:MDS} Let $q=p^m$ and $1\le k\le \lfloor \frac{N+q-1}{q+1} \rfloor$. If one of the following condition holds
\begin{enumerate}[i)]
\item \label{item:i}$(N-1)|(q^2-1)$, 
%\item $N=2n+1$, $m=2s$, $q_0=p^s$ odd, $n= \frac{q-1}{p^r+1}$ even, $n<q_0+1$, $r|\frac{m}{2}$,
\item \label{item:iii} $N=(t+1)n+1$, $n|(q^2-1)$, $n_2=\frac{n}{\gcd (n,q+1)}$, $1\le t\le \frac{q-1}{n_2}-1$,
\item \label{item:iv} $N=tq,1\le t\le q.$
\end{enumerate} 
Then,
\begin{enumerate}
\item there exists an MDS quantum code with parameters $[[N,N-2k,k+1]]_q$;
\item there exists an MDS quantum code with parameters $[[N+1,N-2k-1,k+2]]_q$ if $k(q+1)=N-1$;
\item there exists an MDS quantum code with parameters $[[N+1,N+1-2k,k+2]]_q$ if $N-1+q\not=k(q+1)\not=N-1$.
\end{enumerate}
Moreover, for  $k_0=\lfloor\frac{n+q-1}{q+1}\rfloor$ and $(n-1)|(q^2-1)$, there exists an MDS quantum code with parameters $[[n',n'-2k,k+1]]_q$, where $n'=n$ or $n'=n+1$ and $k\ge k_0+2$, if $(n-1){\not|}k_0(q+1)$, $(n-1){\not|}(k_0+1)(q+1)$ and $2k_0+q\le n$.
\end{thm}

\begin{rem}
%\textnormal{
For length $N=tq$, our codes with parameters $[[N,N-2k,k+1]]_q$ were already obtained in \cite{FangFu}.
For some intermediate lengths, the parameters of our MDS quantum codes from families \ref{item:i})-\ref{item:iii}) are new compared to \cite{FangFu,JinXing14} while the parameters for the maximal length $q^2+1$ were already obtained in  \cite{Ball,FangFu,GraBet,Gra Bet2,GraRot,Gu11,JinXing14,KZ12,Li Xin Wan2,LinLingLuoXin}. To see that our constructions produce new parameters, let us make some comparisons on some code lengths. First recall that the code lengths in \cite{FangFu} are of two possible forms (see Table \ref{table:known}), say $N=tq$ or $N=t(q+1)+2$.
\begin{enumerate}
\item for $q=23$ and $n=67$, we obtain from, Theorem \ref{Q:MDS} \ref{item:i}) 1), a quantum code with parameters $[[ 67, 61, 4 ]]_{23}$, and fromTheorem \ref{Q:MDS} \ref{item:i}) 3), a quantum code with parameters $[[ 67, 59, 5 ]]_{23}$. However, in \cite{JinXing14}, there does not exist any code quantum code of length $67$ over $\F_{23}$ since only the condition $67=3(q-1)+1$ holds but $q\equiv 2\pmod 6$ does not hold (see also Table \ref{table:known}).
The codes in \cite[Theorem 2]{FangFu} have parameters $[[ 69, 63, 4 ]]_{23}$ and $[[ 74, 61, 5 ]]_{23}$, and thus their parameters are worse than ours if the puncturing rules \cite{GraRot} are applied.
\item for $q=17$, $n=12$ and $t=2$,  we obtain, from Theorem \ref{Q:MDS} \ref{item:iii}) 1), a quantum code with parameters $[[25, 21, 3]]_{17}$, and from Theorem \ref{Q:MDS} \ref{item:iii}) 2), a quantum code with parameters $[[26, 20, 4]]_{17}$, while the codes in \cite{JinXing14} can not produce any parameters for this length since, for some positive integer $t$, $t(17+1)+2$ is a bit far from $25$. In \cite{FangFu}, among the code lengths closer to $25$, there are only two quantum codes which have parameters $[[ 20, 13, 3 ]]_{17}$ and $[[ 34, 30, 3 ]]_{17}$. Hence, our codes are new compared to \cite{JinXing14} and \cite{FangFu}. 
%With a similar discussion, we obtain that the codes with parameters $[[31,27,3]]_{19}$, $[[32,26,4]]_{19}$ are all new compared to \cite{JinXing14,FangFu}. However, the parameters $[[25, 21, 3]]_{17}$ and  $[[31,27,3]]_{19}$ were obtained in \cite{LiXu}. 
More parameters of quantum codes over $\F_{17}$ and $\F_{19}$ from Theorems \ref{Q:MDS} \ref{item:iii}) 1)-3) are given in Table \ref{table:MDS}.
\item For $q=7$ and $n=21$, the maximum distance of \cite{FangFu} is just $4$, but our code from Theorem \ref{Q:MDS} \ref{item:iv}) 2) can take $5$, more precisely the code \cite{FangFu} has parameters $[[21,15,4]]_7$ while our parameters are $[[21,13,5]]_7$. It should be noted that the code lengths \cite{JinXing14} are a bit far from $21$. Following a similar discussion, we infer that more parameters of our codes over $\F_{7}$, $\F_8$ and $\F_9$ from Theorem \ref{Q:MDS} \ref{item:iv}) 2)-3) are new compared to \cite{JinXing14,FangFu} (see Table \ref{table:MDS}). 
\item When the maximal dimension $k$ (of the original MDS code with parameters $[n,k]$) satisfies $n-1+q\not= k(q+1)\not= n-1$, our embedding Lemma \ref{lem:embedding} may still be applicable for second time (even more than two times), that is, an MDS Hermitian self-orthogonal $[n,k+1]_{q^2}$ code can be embedded into an MDS Hermitian self-orthogonal code with parameters $[n,k+2]_{q^2}$ or $[n+1,k+2]_{q^2}$ and thus an MDS quantum code with parameters $[[n,n-2(k+2),k+3]]_q$ or $[[n+1,n+1-2(k+2),k+3]]_q$. However, we do not know under which conditions $n,k,q$ that the embedding construction works. We provide some generator matrices of the MDS Hermitian self-orthogonal codes obtained from such an embedding in the appendix, and also the corresponding MDS quantum codes in Table \ref{table:MDS}.
\end{enumerate}
% }
\end{rem}

By applying Lemma \ref{lem:H-construction} to the Hermitian self-orthogonal codes obtained from {\bf Construction \ref{thm:elliptic}}, we obtain the following result.
\begin{thm}\label{thm:Q-elliptic}
Let $q=2^s$ and $2\le k\le \lfloor \frac{2n+1+q}{q+1} \rfloor$. If {\bf Assumption 1} holds, then
there exists a quantum code with parameters $[[2n,2n-2k+2,k-1]]_{q}$, where $n=q^2+2q$.
\end{thm}
\begin{rem} Note that for $q=2^m$ our quantum code lengths ($q^2+2q$) in Theorem \ref{thm:Q-elliptic} are much bigger than those \cite{JinXing12} ($\le q+\lfloor 2\sqrt{q}\rfloor-5$). For $q=4$, the possible parameters \cite{JinXing12} with the maximum length are $[[3,1,2]]_4$ while our parameters are $[[20,12,4]]_4$ which are much better than \cite{JinXing12} and also new compared to \cite{Data-BieEde}. The generator matrix of the classical Hermitian self-orthogonal code, used to construct the $[[20,12,4]]_4$ code, is given in Example \ref{exam:elliptic}. For other parameters compared to \cite{Data-BieEde}, see Table \ref{table:mixed}.
\end{rem}

By applying Lemma \ref{lem:H-construction} to the Hermitian self-orthogonal codes obtained from {\bf Construction \ref{thm:hyper-elliptic2}}, we obtain the following result.
\begin{thm}\label{thm:Q-hyper-elliptic}

Let $q=2^m, m\ge 2$ and $1+q/2\le k\le \lfloor \frac{2N+2q-1}{q+1} \rfloor$. Then

\begin{enumerate}
\item for $N=q^2$, there exists a quantum code with parameters $[[2N,2N-2k+q,\ge k-q+1]]_{q}$ code;
\item for $(N-1)|(q^2-1)$, there exists a quantum code with parameters $[[2N,2N-2k+q,\ge k-q+1]]_{q}$ code.
\end{enumerate}
\end{thm}

\begin{rem}
%\textnormal{
For intermediate lengths, say $2N$ with $N|\frac{q^2-1}{2}$, the code parameters obtained from Theorem \ref{thm:Q-hyper-elliptic} are new compared to \cite{Jin} since \cite{Jin} did not consider these lengths.
%For $q=4$, we deduce, from Example \ref{exam:hyper-elliptic-new}, that there exists a quantum code with parameters $[[31,25,3]]_4$ while the parameters in \cite{Jin} are just $[[32,24,3]]_4$. 
%}
\end{rem}

By applying Lemma \ref{lem:H-construction} to the Hermitian self-orthogonal codes obtained from {\bf Construction \ref{thm:hermitian-2}}, we obtain the following result.
\begin{thm} \label{thm:Q-hermitian-1} Let $q=p^m\ge 4$ and $1+q(q-1)/2\le k\le \lfloor \frac{Nq+q^2-1}{q+1} \rfloor$. Then, there exists a quantum $[[Nq,Nq-2k+q(q-1), \ge k+1-q(q-1) ]]_{q}$ code if one of the following condition holds
\begin{enumerate}[i)]
\item $(N-1)|(q^2-1)$, 
%\item $N=p^r,r|m$,
%\item $N\le q$,
%\item $N=q^2$,
%\item $N=2n+1$, $m=2s$, $q_0=p^s$ odd, $n= \frac{q-1}{p^r+1}$ even, $n<q_0+1$, $r|\frac{m}{2}$,
%\item $N=q^2$,
\item $N=(t+1)n+1$, $n|(q^2-1)$, $n_2=\frac{n}{\gcd (n,q+1)}$, $1\le t\le \frac{q-1}{n_2}-1$,
\item $N=tq$, $1\le t\le q$.
\end{enumerate} 
\end{thm}
\begin{rem} 
%\textnormal{
Theorem \ref{thm:Q-hermitian-1} enables constructions of quantum codes with different lengths. However, for small lengths, the minimum distances of the constructed codes may stay a bit far from the Singleton bound. For $q=5$, we obtain some new quantum codes with parameters $[[95,89,3]]_5$, $[[95,87,3]]_5$, $[[95,83,4]]_5$, and these parameters improve those $[[97,87,3]]_5$, $[[96,82,4]]_5$ of the good quantum codes in the database \cite{Data-BieEde}.
%}
\end{rem}
By applying Lemma \ref{lem:H-construction} to the Hermitian self-orthogonal codes obtained from {\bf Construction \ref{thm:hermitian-3}}, we obtain the following result.
\begin{thm}\label{thm:Q-hermitian-2} Let $q=p^m\ge 4$ and $1+\frac{(q-1)^2}{4}\le k\le \lfloor \frac{s+(q-1)^2/2+q-1}{q+1} \rfloor$. Then, there exists a quantum $[[s,s-2k+\frac{(q-1)^2}{2}, \ge k+1-\frac{(q-1)^2}{2} ]]_{q}$ code.
\end{thm}
\begin{rem}
%\textnormal{
 By using Magma, we obtain, from Theorem \ref{thm:Q-hermitian-2}, a quantum code with parameters $[[15,9,3]]_3$, and these parameters are the same as those in \cite{Data-BieEde}. For $q$ small, our Theorem \ref{thm:Q-hermitian-2} produces many good quantum codes due to the fact that the classical Hermitian self-orthogonal codes and their dual codes have good minimum distances in the sense that the minimum distances are closer to the Singleton bound than those in Theorem \ref{thm:Q-hermitian-1}, for instance, we obtain a $[[ 175, 169,  3 ]]_7$ code while the code \cite{Data-BieEde} has parameters $[[171,169,2]]_7$ (see Table \ref{table:mixed}).
%}
\end{rem}
By applying Lemma \ref{lem:H-construction} to the Hermitian self-orthogonal codes obtained from {\bf Constructions \ref{thm:other-curves}-\ref{thm:other-curves-even}}, we obtain the following result.
\begin{thm} \label{thm:Q-other-curves} Let $q=p^m\ge 4$ and $t$ be a positive integer such that $\gcd(t,(q+1))|\frac{q+1}{2}$. Put $s=\gcd(t(q-1),(q+1)(q-1))$ and $1+(t-1)(q-1)/2\le k\le \lfloor \frac{(s+1)q+t(q-1)}{q+1} \rfloor$. Then, there exists a quantum $[[(s+1)q,(s+1)q-2k+(t-1)(q-1), \ge k+1-(t-1)(q-1) ]]_{q}$ code if one of the following condition holds:
\begin{enumerate}[i)]
\item $q$ is odd, and $t$ is any positive integer;
\item $q$ is even, and $t$ is an odd positive integer.
\end{enumerate} 
\end{thm}
\begin{rem} We obtain some quantum codes whose parameters improve those in \cite{Data-BieEde}, for instance,
${ [[91,81,4]]_7}$, ${ [[176,168,3]]_8}$, ${ [[369,361,3]]_9}$ codes while the parameters \cite{Data-BieEde} are $[[176,165,3]]_8$, $[[90,80,3]]_7$, $[[176,165,3]]_8$, $[[381,361,3]]_9$.
\end{rem}
\begin{table}
\centering
\caption{Some MDS quantum codes, $^*:$ new compared to \cite{FangFu,GraBet,Gra Bet2,GraRot,JinXing14,LiXu}, $^+:$ obtained in the implicit constructions \cite{GraRot} but new compared to the explicit constructions \cite{FangFu,JinXing14}
%, $^+:$ obtained in \cite{JinXing14},$^\times:$ obtained in \cite{FangFu}
\label{table:MDS}
}
$$
\begin{array}{c|c}
\hline
\hline

\text{Theorem \ref{Q:MDS} \ref{item:i}) 2)-3)}&\text{Theorem \ref{Q:MDS-new}}\\

\hline
\hline

%{ [[ 67, 61, 4 ]]_{23}^*}&{ [[ 67, 59, 5 ]]_{23}^*}\\

[[9,3,4]]_{5}^+&[[9,3,4]]_{5}^+\\

[[13,7,4]]_{7}^+&[[14,6,5]]_{7}^+\\

[[17,11,4]]_{9}^+&[[17,7,6]]_{9}^+\\

[[21,15,4]]_{11}^+&[[22,10,7]]_{11}^+\\

[[25,19,4]]_{13}^+&[[25,11,8]]_{13}^*\\

[[33,27,4]]_{17}^+&[[33,15,10]]_{17}^*\\

[[37,31,4]]_{19}^+&[[38,18,11]]_{19}^*\\

[[45,39,4]]_{23}^+&[[46,22,13]]_{23}^*\\

[[49,43,4]]_{25}^+&[[49,23,14]]_{25}^*\\

[[53,47,4]]_{27}^+&[[54,26,15]]_{27}^*\\

\hline
\hline

\text{Theorem \ref{Q:MDS} \ref{item:iii}) 1)}, n=12&\text{Theorem \ref{Q:MDS} \ref{item:iii}) 2)-3)}, n=12\\

\hline
\hline

%{ [[25, 21, 3]]_{17}^+}&{ [[26, 20, 4]]_{17}^+}\\

%{ [[37, 33, 3]]_{17}^+}&{ [[38, 32,4]]_{17}^+}\\

{ [[49, 43, 4]]_{17}^+}&{ [[49,41,5]]_{17}^+}\\

{ [[61, 53, 5]]_{17}^+}&{ [[62, 52,6]]_{17}^+}\\

{ [[73, 65,5]]_{17}^+}&{ [[74, 64, 6]]_{17}^+}\\

{ [[85,7 5, 6]]_{17}^+}&{ [[85, 73,7]]_{17}^+}\\

{ [[97, 85,7]]_{17}^+}&{ [[97, 83,8]]_{17}^+}\\

\hline
\hline

\hline
\hline
\text{Theorem \ref{Q:MDS} \ref{item:iv}) 2)-3)}&\text{Subcode embedding}\\
\hline
\hline
{ [[21, 13, 5]]_7^+}&[[22, 12, 6]]_7^+\\
{ [[28, 18, 6]]_7^+}&-\\
{ [[36, 24, 7]]_7^+}&-\\
{ [[16, 10, 4]]_8^+}&[[16, 8,5]]_8^+,[[16, 6,6]]_8^+\\
{ [[24, 16, 5]]_8^+}&[[24, 14,6]]_8^+,[[24, 12,7]]_8^+\\
{ [[32, 22,6]]_8^+}&[[32, 20,7]_8^+\\
{ [[40, 28,7]]_8^+}&-\\
{ [[48, 34, 8]]_8^+}&-\\
{ [[18, 12,4]]_9 ^+}&[[18, 10,5]]_9^+,[[18, 8,6]]_9^+\\
{ [[27, 19,5]]_9^+}&[[27, 17,6]]_9^+,[[28, 16,7]]_9^+\\
{ [[36, 26,6]]_9^+}&[[36, 24,7]]_9^+\\
{ [[45, 33,7]]_9^+}&[[45, 31,8]]_9^+\\
{ [[54, 40,8]]_9^+}&[[55, 39,9]]_9^+\\
{ [[64, 48,9]]_9^+}&-\\

\end{array}
$$
\label{table:mds1}
\end{table}

\begin{table}
\centering
\caption{Some quantum codes, $^*$: new compared to \cite{Data-BieEde,Jin,JinXing12}
}
{
$$
\begin{array}{c|c}
\hline
\hline
\text{Theorem \ref{thm:Q-elliptic}}&\text{Parameters } \cite{Data-BieEde}\\
\hline
\hline
[[24,18,3]]_4^*&[[22,17,3]]_4\\

[[20,12,4]]_4^*&[[21,13,4]]_4\\

[[80,64,8]]_8^*&[[85,45,8]]_8\\

[[288,228,30]]_{2^4}^*&-\\

[[1088,884,102]]_{2^5}^*&-\\
\hline
\hline
\text{Theorem \ref{thm:Q-hermitian-1}}&\text{Parameters } \cite{Data-BieEde}\\
\hline
\hline

[[64,58,3]]_4&[[63,58,3]]_4\\

{ [[95, 89, 3]]_5^*}&[[97,87,3]]_5\\

{ [[95, 87, 3]]_5^*}&[[97,87,3]]_5\\

[[95, 85, 3]]_5^*&-\\

{ [[95, 83, 4]]_5^*}&[[96,82,4]]_5\\

\hline
\hline
\text{Theorem \ref{thm:Q-hermitian-2}}&\text{Parameters } \cite{Data-BieEde}\\
\hline
\hline
{[[15,9,3]]_3}&[[15,9,3]]_3\\

{ [[ 65, 59, 3 ]]_5^*}&[[63,58,3]]_5\\

{ [[65,51,5]]_5^*}&-\\

[[65,49,\ge 5]]_5&[[63,49,6]]_5\\

{ [[ 175, 169,  3 ]]_7^*}&[[171,169,2]]_7\\
\hline
\hline
\text{Theorem \ref{thm:Q-other-curves}}&\text{Parameters } \cite{Data-BieEde}\\
\hline
\hline

[64, 58, 3]_{4}&[64, 59, 3]_{4}\\

{ [[91,81,4]]_7^*}&[[90,80,3]]_7\\

{ [[176,168,3]]_8^*}&[[176,165,3]]_8\\

[[63,55,3]]_9&[[62,55,3]]_9\\

{ [[369,361,3]]_9^*}&[[381,361,3]]_9\\
\end{array}
$$
}
\label{table:mixed}
\end{table}

\section*{Appendix}
%1 1 1 1 1 1 1 1 1 1 1 1 1 1 1 1 1 1 1 1 1 1 {0}\\
%0 1 w^{3 }w^{6 }w^{9 }w^{12 }w^{15 }w^{18 }w^{21 }w^{24 }w^{27 }w^{30 }w^{33 }w^{36 }w^{39 }w^{42 }w^{45}w^{48 }w^{51 }w^{54 }w^{57 }w^{60 }{0}\\
%0 1 w^{6 }w^{12 }w^{18 }w^{24 }w^{30 }w^{36 }w^{42 }w^{48 }w^{54 }w^{60 }w^{3 }w^{9 }w^{15 }w^{21 }w^{27}w^{33 }w^{39 }w^{45 }w^{51 }w^{57 }{0}\\
%0 1 w^{9 }w^{18 }w^{27 }w^{36 }w^{45 }w^{54}{  1 }w^{9 }w^{18 }w^{27 }w^{36 }w^{45 }w^{54}{ 1 }w^{9 }w^{18 }w^{27}w^{36 }w^{45 }w^{54 }{0}\\
%0 1 w^{12 }w^{24 }w^{36 }w^{48 }w^{60 }w^{9 }w^{21 }w^{33 }w^{45 }w^{57 }w^{6 }w^{18 }w^{30 }w^{42 }w^{54}w^{3 }w^{15 }w^{27 }w^{39 }w^{51 }{0}\\
%23, 5\\ Linear Code over GF(2^6)
%[22, 5, 18]_{7^2}
The following matrices are the generator matrices of the MDS Hermitian self-orthogonal $[n,k]_{q^2}$ codes that give rise to MDS quantum codes in Table \ref{table:MDS}). The matrices are written in the systematic form $(I|A_{\F_{q^2}})$ for embedding the subcodes in Theorem \ref{Q:MDS} \ref{item:iv}) 2)-3) and in the systematic form $(I|A_{\F_{q^2}}|B_{\F_{q^2}})$ for Theorem \ref{Q:MDS-new}, where the identity matrix  $I$ is omitted.

\noindent
\begin{table*}[th]

{\tiny
$
\left(
\begin{array}{lllllllllllllllllllll}
\theta^{35 }&\theta^{35 }&\theta^{37 }&\theta^{2 }&\theta^{14 }&\theta^{44 }&\theta^{19 }&\theta^{46 }&\theta^{42 }&\theta^{37 }&\theta^{18}&{ \theta }&\theta^{26 }&\theta^{10 }&\theta^{4}&\theta^{2}&{  4}\\
\theta^{6 }&\theta^{4 }&\theta^{5 }&\theta^{12 }&\theta^{9 }&\theta^{28 }&\theta^{6 }&\theta^{5 }&\theta^{18 }&\theta^{15 }&\theta^{26}&{  1 }&\theta^{23 }&\theta^{35 }&\theta^{46}&\theta^{38}&{  4}\\
\theta^{19}&{  3}&{  3 }&\theta^{20 }&\theta^{47 }&{ \theta }&\theta^{17 }&\theta^{6 }&\theta^{11 }&\theta^{37 }&\theta^{42 }&\theta^{36 }&\theta^{6 }&\theta^{11 }&\theta^{3 }&\theta^{27}&\theta^{34}\\
\theta^{26}&{  2 }&\theta^{6 }&\theta^{36 }&\theta^{20 }&\theta^{37 }&\theta^{27 }&\theta^{12 }&\theta^{3 }&\theta^{34 }&\theta^{19}&{  2 }&\theta^{15 }&\theta^{33 }&\theta^{3}&\theta^{12 }&\theta^{9}\\
\theta^{45}&{  4 }&\theta^{11 }&\theta^{13}&{  3 }&\theta^{17 }&\theta^{21 }&\theta^{41 }&\theta^{29 }&\theta^{36 }&\theta^{30 }&\theta^{17}&{  3 }&\theta^{30 }&\theta^{42}&{  2}&{ 6}\\
\end{array}
\right)_{\F_{7^2}},
$
\\
%17, 4\\ Linear Code over GF(2^6)
$
\left(
\begin{array}{lllllllllllllllll}
\theta^{3 }&\theta^{14 }&\theta^{32 }&\theta^{7 }&\theta^{34 }&\theta^{33 }&\theta^{40 }&\theta^{50 }&\theta^{53 }&\theta^{22 }&\theta^{31 }&\theta^{51 }\\
\theta^{58 }&\theta^{11 }&\theta^{49 }&\theta^{32 }&\theta^{49 }&\theta^{60 }&\theta^{22 }&\theta^{37 }&\theta^{4 }&\theta^{35 }&\theta^{39 }&\theta^{16 }\\
\theta^{19 }&\theta^{45 }&\theta^{39 }&\theta^{31 }&\theta^{14 }&\theta^{43 }&\theta^{38 }&\theta^{25 }&\theta^{15 }&\theta^{4 }&\theta^{48 }&\theta^{54 }\\
\theta^{24 }&\theta^{13 }&\theta^{35 }&\theta^{22 }&\theta^{28 }&\theta^{4 }&\theta^{7 }&\theta^{31 }&\theta^{38 }&\theta^{22 }&\theta^{4 }&\theta^{4 }\\
\end{array}
\right)_{\F_{8^2}},
$
\\
%17, 5\\ Linear Code over GF(2^6)
$
\left(
\begin{array}{lllllllllllllllll}
\theta^{27 }&\theta^{59 }&\theta^{10 }&\theta^{25 }&\theta^{11 }&\theta^{49 }&\theta^{31 }&\theta^{35 }&\theta^{2 }&\theta^{55 }&\theta^{25 }\\
\theta^{58 }&\theta^{47 }&\theta^{6 }&\theta^{11 }&\theta^{9 }&\theta^{2 }&\theta^{52 }&\theta^{20 }&\theta^{49 }&\theta^{34 }&\theta^{24 }\\
\theta^{22 }&\theta^{30 }&\theta^{61 }&\theta^{32 }&\theta^{48 }&\theta^{11 }&\theta^{33 }&\theta^{24 }&\theta^{11 }&\theta^{36 }&\theta^{55 }\\
\theta^{8 }&\theta^{44 }&\theta^{7}&{  \theta }&\theta^{27 }&\theta^{61 }&\theta^{57 }&\theta^{2 }&\theta^{47 }&\theta^{10 }&\theta^{23 }\\
\theta^{14 }&\theta^{38 }&\theta^{17 }&\theta^{49 }&\theta^{54}&{  \theta }&\theta^{13 }&\theta^{23 }&\theta^{56 }&\theta^{6 }&\theta^{28 }\\
\end{array}
\right)_{\F_{8^2}},
$
\\
%25, 5\\ Linear Code over GF(2^6)
$
\left(
\begin{array}{lllllllllllllllllllllllll}
\theta^{24 }&\theta^{52 }&\theta^{3 }&\theta^{18 }&\theta^{31 }&\theta^{30 }&\theta^{37 }&\theta^{33 }&\theta^{25 }&\theta^{30 }&\theta^{60 }&\theta^{12 }&\theta^{58 }&\theta^{61}&\theta^{6 }&\theta^{6 }&\theta^{57 }&\theta^{19 }&\theta^{48 }\\
\theta^{31 }&\theta^{40 }&\theta^{22 }&\theta^{58 }&\theta^{13 }&\theta^{32 }&\theta^{15 }&\theta^{43 }&\theta^{25 }&\theta^{53 }&\theta^{9 }&\theta^{42 }&\theta^{30 }&\theta^{44}&\theta^{5 }&\theta^{4 }&\theta^{37 }&\theta^{12 }&\theta^{61 }\\
\theta^{13 }&\theta^{43 }&\theta^{38 }&\theta^{35}&{  1 }&\theta^{38 }&\theta^{52 }&\theta^{58 }&\theta^{6 }&\theta^{9 }&\theta^{8 }&\theta^{11 }&\theta^{34 }&\theta^{62 }&\theta^{32}&\theta^{52 }&\theta^{39 }&\theta^{37 }&\theta^{52 }\\
\theta^{60 }&\theta^{44 }&\theta^{60 }&\theta^{52 }&\theta^{50 }&\theta^{43 }&\theta^{15 }&\theta^{51 }&\theta^{49 }&\theta^{45 }&\theta^{7 }&\theta^{20 }&\theta^{12 }&\theta^{51}&\theta^{2 }&\theta^{35 }&\theta^{43 }&\theta^{45 }&\theta^{28 }\\
\theta^{15 }&\theta^{17 }&\theta^{10 }&\theta^{3 }&\theta^{8 }&\theta^{44 }&\theta^{42 }&\theta^{28 }&\theta^{35 }&\theta^{36 }&\theta^{29 }&\theta^{39 }&\theta^{54 }&\theta^{60}&\theta^{28 }&\theta^{41 }&\theta^{43 }&\theta^{30 }&\theta^{15 }\\
\end{array}
\right)_{\F_{8^2}},
$
\\
%25, 6\\ Linear Code over GF(2^6)
$
\left(
\begin{array}{lllllllllllllllllllllllll}
\theta^{45 }&\theta^{52 }&\theta^{41 }&\theta^{47 }&\theta^{24 }&\theta^{48 }&\theta^{11 }&\theta^{6 }&\theta^{5}&{  \theta }&\theta^{45 }&\theta^{46 }&\theta^{27 }&\theta^{4}&\theta^{25 }&\theta^{59 }&\theta^{31 }&\theta^{3 }\\
\theta^{34 }&\theta^{9 }&\theta^{19 }&\theta^{30 }&\theta^{27 }&\theta^{27 }&\theta^{22 }&\theta^{7 }&\theta^{29 }&\theta^{14 }&\theta^{13 }&\theta^{19 }&\theta^{11 }&\theta^{4}&\theta^{24 }&\theta^{40 }&\theta^{25 }&\theta^{17 }\\
\theta^{35 }&\theta^{23 }&\theta^{57 }&\theta^{15 }&\theta^{31 }&\theta^{62 }&\theta^{35 }&\theta^{49 }&\theta^{46 }&\theta^{11 }&\theta^{43 }&\theta^{21 }&\theta^{27}&\theta^{29 }&\theta^{7 }&\theta^{40 }&\theta^{48 }&\theta^{6 }\\
\theta^{17 }&\theta^{26 }&\theta^{55 }&\theta^{46 }&\theta^{17 }&\theta^{6 }&\theta^{9 }&\theta^{10}&{  1 }&\theta^{54 }&\theta^{33 }&\theta^{43 }&\theta^{60 }&\theta^{43}&\theta^{34 }&\theta^{25 }&\theta^{37 }&\theta^{26 }\\
\theta^{26 }&\theta^{12 }&\theta^{42 }&\theta^{40 }&\theta^{54 }&\theta^{6 }&\theta^{22 }&\theta^{32 }&\theta^{27 }&\theta^{49 }&\theta^{25 }&\theta^{58 }&\theta^{42}&\theta^{42 }&\theta^{13 }&\theta^{61 }&\theta^{58 }&\theta^{49 }\\
\theta^{29 }&\theta^{44 }&\theta^{5 }&\theta^{21 }&\theta^{26 }&\theta^{36 }&\theta^{33 }&\theta^{30 }&\theta^{39 }&\theta^{8 }&\theta^{26 }&\theta^{11 }&\theta^{18 }&\theta^{29}&\theta^{30 }&\theta^{19 }&\theta^{47 }&\theta^{15 }\\
\end{array}
\right)_{\F_{8^2}},
$
\\
\\
%19, 4\\ Linear Code over GF(3^4)
$
\left(
\begin{array}{lllllllllllllllllllllllll}
\theta^{49 }&\theta^{52 }&\theta^{26 }&\theta^{30 }&\theta^{24 }&\theta^{28 }&\theta^{65 }&\theta^{16 }&\theta^{43 }&\theta^{79 }&\theta^{63 }&\theta^{8}&{ \theta }&\theta^{6 }\\
\theta^{42 }&\theta^{58 }&\theta^{9 }&\theta^{63 }&\theta^{49 }&\theta^{41 }&\theta^{62 }&\theta^{59 }&\theta^{46 }&\theta^{75 }&\theta^{27 }&\theta^{31 }&\theta^{41 }&\theta^{16 }\\
\theta^{41 }&\theta^{59 }&\theta^{19 }&\theta^{67 }&\theta^{42 }&\theta^{33 }&\theta^{46 }&\theta^{50 }&\theta^{54 }&\theta^{63 }&\theta^{9 }&\theta^{8 }&\theta^{41 }&\theta^{31 }\\
\theta^{16 }&\theta^{16 }&\theta^{55 }&\theta^{71 }&\theta^{49 }&\theta^{73 }&\theta^{43 }&\theta^{44 }&\theta^{57 }&\theta^{30 }&\theta^{26 }&\theta^{78 }&\theta^{28 }&\theta^{8 }\\
\end{array}
\right)_{\F_{9^2}},
$
\\
%19, 5\\ Linear Code over GF(3^4)
$
\left(
\begin{array}{lllllllllllllllllllllllll}
\theta^{6 }&\theta^{56}&{  1 }&\theta^{72 }&\theta^{38 }&\theta^{6 }&\theta^{56 }&\theta^{63 }&\theta^{32 }&\theta^{78 }&\theta^{78 }&\theta^{27 }&\theta^{78 }\\
\theta^{32 }&\theta^{59 }&\theta^{53 }&\theta^{37 }&\theta^{71 }&\theta^{23 }&\theta^{39 }&\theta^{6 }&\theta^{48 }&\theta^{62 }&\theta^{41 }&\theta^{7 }&\theta^{28 }\\
\theta^{76 }&\theta^{32 }&\theta^{20 }&\theta^{73 }&\theta^{26 }&\theta^{50 }&\theta^{73 }&\theta^{57 }&\theta^{79 }&\theta^{7 }&\theta^{61 }&\theta^{50 }&\theta^{6 }\\
\theta^{61 }&\theta^{16 }&\theta^{52 }&\theta^{28 }&\theta^{14 }&\theta^{75 }&\theta^{15 }&\theta^{8 }&\theta^{74 }&\theta^{52 }&\theta^{79 }&\theta^{65 }&\theta^{11 }\\
\theta^{11 }&\theta^{67 \theta }&\theta^{4 }&\theta^{9 }&\theta^{48 }&\theta^{7 }&\theta^{44 }&\theta^{5 }&\theta^{71 }&\theta^{19 }&\theta^{13 }&\theta^{23 }\\
\end{array}
\right)_{\F_{9^2}},
$
\\
%28, 5\\ Linear Code over GF(3^4)
$
\left(
\begin{array}{lllllllllllllllllllllllll}
\theta^{71 }&\theta^{62 }&\theta^{22 }&\theta^{51 }&\theta^{39 }&\theta^{31 }&\theta^{41 }&\theta^{78 }&\theta^{74 }&\theta^{55 }&\theta^{7 }&\theta^{67 }&\theta^{28 }&\theta^{74}&\theta^{14}&{  1 }&\theta^{49 }&\theta^{2 }&\theta^{55 }&\theta^{74 }&\theta^{54 }&\theta^{74 }\\
\theta^{59 }&\theta^{59 }&\theta^{43 }&\theta^{36 }&\theta^{49 }&\theta^{18 }&\theta^{45 }&\theta^{35 }&\theta^{23 }&\theta^{72 }&\theta^{8 }&\theta^{41 }&\theta^{75 }&\theta^{10}&\theta^{21}&{  1 }&\theta^{17 }&\theta^{29 }&\theta^{19 }&\theta^{32 }&\theta^{68 }&\theta^{20 }\\
\theta^{8 }&\theta^{9 }&\theta^{16 }&\theta^{6 }&\theta^{14 }&\theta^{76 }&\theta^{60 }&\theta^{10 }&\theta^{7 }&\theta^{18 }&\theta^{51 }&\theta^{37 }&\theta^{16 }&\theta^{17}&\theta^{53 }&\theta^{8 }&\theta^{74 }&\theta^{28 }&\theta^{16 }&\theta^{30 }&\theta^{26 }&\theta^{22 }\\
\theta^{33 }&\theta^{75 }&\theta^{72 }&\theta^{7 }&\theta^{67 }&\theta^{45 }&\theta^{38 }&\theta^{56 }&\theta^{33}&{  \theta }&\theta^{9 }&\theta^{9 }&\theta^{3 }&\theta^{21 }&\theta^{46}&\theta^{5 }&\theta^{16 }&\theta^{23 }&\theta^{36 }&\theta^{70 }&\theta^{20 }&\theta^{69 }\\
\theta^{59 }&\theta^{4 }&\theta^{69 }&\theta^{2 }&\theta^{58 }&\theta^{35 }&\theta^{14 }&\theta^{14 }&\theta^{74 }&\theta^{75}&{  2 }&\theta^{72 }&\theta^{31 }&\theta^{64 }&\theta^{3}&\theta^{6 }&\theta^{67 }&\theta^{47 }&\theta^{57 }&\theta^{72 }&\theta^{31 }&\theta^{21 }\\
\end{array}
\right)_{\F_{9^2}},
$
\\
%28, 6\\ Linear Code over GF(3^4)
$
\left(
\begin{array}{lllllllllllllllllllllllll}
\theta^{25 }&\theta^{6 }&\theta^{5 }&\theta^{8 }&\theta^{77 }&\theta^{65 }&\theta^{20 }&\theta^{49 }&\theta^{63 }&\theta^{74 }&\theta^{71 }&\theta^{19 }&\theta^{8 }&\theta^{9}&\theta^{51 }&\theta^{75 }&\theta^{9 }&\theta^{57 }&\theta^{68 }&\theta^{34 }&\theta^{38 }&\theta^{72}\\
\theta^{33 }&\theta^{38}&{  \theta }&\theta^{29 }&\theta^{75}&{  1 }&\theta^{68 }&\theta^{9 }&\theta^{11 }&\theta^{6 }&\theta^{56 }&\theta^{77 }&\theta^{35 }&\theta^{27 }&\theta^{62}&\theta^{54 }&\theta^{47 }&\theta^{32 }&\theta^{37 }&\theta^{59 }&\theta^{75 }&\theta^{60}\\
\theta^{78 }&\theta^{26 }&\theta^{66 }&\theta^{9 }&\theta^{68 }&\theta^{30 }&\theta^{58 }&\theta^{8 }&\theta^{52 }&\theta^{64 }&\theta^{67 }&\theta^{33 }&\theta^{57 }&\theta^{74}&\theta^{5 }&\theta^{46 }&\theta^{61 }&\theta^{44 }&\theta^{50 }&\theta^{32 }&\theta^{12 }&\theta^{9}\\
\theta^{50 }&\theta^{68 }&\theta^{53 }&\theta^{48 }&\theta^{23 }&\theta^{74 }&\theta^{10 }&\theta^{20 }&\theta^{21 }&\theta^{8 }&\theta^{25 }&\theta^{6 }&\theta^{47 }&\theta^{53}&\theta^{68 }&\theta^{54 }&\theta^{42 }&\theta^{50 }&\theta^{76 }&\theta^{12 }&\theta^{45 }&\theta^{34}\\
\theta^{24 }&\theta^{30 }&\theta^{13 }&\theta^{4 }&\theta^{58 }&\theta^{15 }&\theta^{13 }&\theta^{26 }&\theta^{60 }&\theta^{4 }&\theta^{53 }&\theta^{79 }&\theta^{55 }&\theta^{55}&\theta^{34 }&\theta^{70 }&\theta^{31 }&\theta^{36 }&\theta^{43 }&\theta^{68 }&\theta^{42 }&\theta^{60}\\
\theta^{75 }&\theta^{38 }&\theta^{68 }&\theta^{60 }&\theta^{54 }&\theta^{72 }&\theta^{31 }&\theta^{36 }&\theta^{18 }&\theta^{52 }&\theta^{39}&{  \theta }&\theta^{49 }&\theta^{77}&\theta^{68 }&\theta^{39 }&\theta^{75 }&\theta^{49 }&\theta^{72 }&\theta^{53 }&\theta^{74 }&\theta^{41}\\
\end{array}
\right)_{\F_{9^2}},
$
}
\end{table*}

\begin{table*}
%\caption{The following matrices are the generator matrices of the MDS Hermitian self-orthogonal $[n,k]_{q^2}$ codes in Theorem \ref{Q:MDS-new} of Table \ref{table:MDS}. The matrices are written in the systematic form $(I|A_{\F_{q^2}}|B_{\F_{q^2}})$ with the identity matrix $I$ being omitted}
{\tiny
%[33, 9] Linear Code over GF(17^{2)
$
A_{\F_{17^2}}=\left(
\begin{array}{lllllllllllllllll}
\theta^{255 }&\theta^{53 }&\theta^{17 }&\theta^{254 }&\theta^{236 }&\theta^{23 }&\theta^{41 }&\theta^{91 }&\theta^{253 }&\theta^{39 }&\theta^{183 }&\theta^{10}\\
{8 }&\theta^{161 }&\theta^{266}&{ \theta }&\theta^{269 }&\theta^{276 }&\theta^{200 }&\theta^{188 }&\theta^{88 }&\theta^{145}&{ 13 }&\theta^{2}\\ 
\theta^{224 }&\theta^{160 }&\theta^{19 }&\theta^{109 }&\theta^{165 }&\theta^{238 }&\theta^{188}&{ 16 }&\theta^{244 }&\theta^{56 }&\theta^{183 }&\theta^{81}\\
\theta^{128 }&\theta^{184 }&\theta^{286 }&\theta^{130 }&\theta^{253 }&\theta^{114 }&\theta^{130 }&\theta^{112}&{ 8 }&\theta^{192 }&\theta^{74 }&\theta^{172}\\ 
\theta^{52 }&\theta^{275 }&\theta^{209 }&\theta^{8 }&\theta^{173 }&\theta^{101 }&\theta^{193 }&\theta^{241 }&\theta^{47 }&\theta^{27 }&\theta^{109 }&\theta^{250}\\ 
\theta^{135 }&\theta^{246 }&\theta^{59 }&\theta^{266 }&\theta^{98 }&\theta^{68 }&\theta^{227 }&\theta^{63 }&\theta^{223 }&\theta^{229 }&\theta^{279 }&\theta^{44}\\ 
\theta^{259 }&\theta^{88 }&\theta^{77 }&\theta^{163 }&\theta^{115 }&\theta^{40 }&\theta^{241}&{ 16 }&\theta^{92 }&\theta^{164 }&\theta^{240 }&\theta^{261}\\ 
\theta^{254 }&\theta^{111 }&\theta^{106 }&\theta^{80 }&\theta^{199 }&\theta^{244 }&\theta^{112 }&\theta^{57}&{ 13 }&\theta^{220 }&\theta^{74 }&\theta^{121}\\ 
\theta^{269 }&\theta^{86 }&\theta^{109 }&\theta^{89 }&\theta^{96 }&\theta^{20 }&\theta^{8 }&\theta^{196 }&\theta^{253}&{ 8 }&\theta^{110 }&\theta^{223}\\ 
\end{array}
\right),
$
\\
$
B_{\F_{17^2}}=\left(
\begin{array}{lllllllllllllllll}
\theta^{46 }&\theta^{3 }&\theta^{219 }&\theta^{217 }&\theta^{127 }&\theta^{149 }&\theta^{203 }&\theta^{200 }&\theta^{2 }&\theta^{125 }&\theta^{233 }&\theta^{219}\\
\theta^{115 }&\theta^{83}&{ 12 }&\theta^{100 }&\theta^{124 }&\theta^{17 }&\theta^{110 }&\theta^{267 }&\theta^{53 }&\theta^{154 }&\theta^{212 }&\theta^{188}\\
\theta^{220 }&\theta^{254 }&\theta^{193}&{ 11 }&\theta^{192 }&\theta^{40 }&\theta^{253}&{ \theta }&\theta^{251 }&\theta^{70 }&\theta^{156 }&\theta^{92}\\
\theta^{279 }&\theta^{51 }&\theta^{56 }&\theta^{65}&{ 7 }&\theta^{88 }&\theta^{256 }&\theta^{124 }&\theta^{253 }&\theta^{248 }&\theta^{52 }&\theta^{16}\\
\theta^{269 }&\theta^{9 }&\theta^{40 }&\theta^{115}&{ 9 }&\theta^{281 }&\theta^{203 }&\theta^{26 }&\theta^{275 }&\theta^{149 }&\theta^{129 }&\theta^{99}\\
\theta^{106 }&\theta^{46 }&\theta^{45 }&\theta^{146 }&\theta^{133 }&\theta^{166 }&\theta^{155 }&\theta^{20 }&\theta^{224 }&\theta^{218 }&\theta^{77 }&\theta^{223}\\
\theta^{235 }&\theta^{218 }&\theta^{129}&{ 7 }&\theta^{211 }&\theta^{22 }&\theta^{87 }&\theta^{19 }&\theta^{265 }&\theta^{214 }&\theta^{193 }&\theta^{218}\\
\theta^{63 }&\theta^{246 }&\theta^{200 }&\theta^{181}&{ 14 }&\theta^{287 }&\theta^{130 }&\theta^{138 }&\theta^{163 }&\theta^{154 }&\theta^{88 }&\theta^{233}\\
\theta^{191}&{ 10 }&\theta^{208 }&\theta^{232 }&\theta^{125 }&\theta^{218 }&\theta^{87 }&\theta^{161 }&\theta^{262 }&\theta^{32 }&\theta^{8}&{ 15}\\
 \end{array}
\right),
$  
 \\   
%[38, 10] Linear Code over GF(19^{2)
$
A_{\F_{19^2}}=
\left(
\begin{array}{lllllllllllllllll}
\theta^{248 }&\theta^{71 }&\theta^{159 }&\theta^{10 }&\theta^{264 }&\theta^{88 }&\theta^{169 }&\theta^{107}&{ 12 }&\theta^{175 }&\theta^{298 }&\theta^{290 }&\theta^{241 }&\theta^{227}\\ 
\theta^{90 }&\theta^{308 }&\theta^{258 }&\theta^{287 }&\theta^{317 }&\theta^{205 }&\theta^{195 }&\theta^{284 }&\theta^{134 }&\theta^{232 }&\theta^{142 }&\theta^{127 }&\theta^{297 }&\theta^{24}\\ 
\theta^{28 }&\theta^{116 }&\theta^{239 }&\theta^{312 }&\theta^{24 }&\theta^{344 }&\theta^{129 }&\theta^{336 5 }&\theta^{212 }&\theta^{198 }&\theta^{337 }&\theta^{281 }&\theta^{126}\\ 
\theta^{348 }&\theta^{83 }&\theta^{76 }&\theta^{322 }&\theta^{78}&{ 16 }&\theta^{297 }&\theta^{299 }&\theta^{41 }&\theta^{67 }&\theta^{207 }&\theta^{62}&{ 9 }&\theta^{139}\\ 
\theta^{151 }&\theta^{342 }&\theta^{342 }&\theta^{98 }&\theta^{27 }&\theta^{73 }&\theta^{332 }&\theta^{46 }&\theta^{303 }&\theta^{87}&{ \theta }&\theta^{10 }&\theta^{184 }&\theta^{317}\\ 
\theta^{131}&{ 7 }&\theta^{216 }&\theta^{339 }&\theta^{138 }&\theta^{357}&{ 12 }&\theta^{56 }&\theta^{25 }&\theta^{324 }&\theta^{356 }&\theta^{139 }&\theta^{107 }&\theta^{316}\\ 
\theta^{261 }&\theta^{78 }&\theta^{332 }&\theta^{191 }&\theta^{357 }&\theta^{86 }&\theta^{202 }&\theta^{2 }&\theta^{13 }&\theta^{24 }&\theta^{211 }&\theta^{112 }&\theta^{214 }&\theta^{217}\\ 
\theta^{181 }&\theta^{183 }&\theta^{265 }&\theta^{282 }&\theta^{184}&{ 6 }&\theta^{266 }&\theta^{239 }&\theta^{294 }&\theta^{347 }&\theta^{246 }&\theta^{302 }&\theta^{162 }&\theta^{299}\\ 
\theta^{278 }&\theta^{42 }&\theta^{309 }&\theta^{154 }&\theta^{214 }&\theta^{46 }&\theta^{39 }&\theta^{242 }&\theta^{110 }&\theta^{207 }&\theta^{148 }&\theta^{276 }&\theta^{291 }&\theta^{186}\\ 
\theta^{218 }&\theta^{168 }&\theta^{197 }&\theta^{227 }&\theta^{115 }&\theta^{105 }&\theta^{194 }&\theta^{44 }&\theta^{142 }&\theta^{52 }&\theta^{37 }&\theta^{207 }&\theta^{294 }&\theta^{344}\\ 
\end{array}
\right),
$
\\
$
B_{\F_{19^2}}=\left(
\begin{array}{lllllllllllllllll}
\theta^{331 }&\theta^{110 }&\theta^{208 }&\theta^{175 }&\theta^{30 }&\theta^{287 }&\theta^{79 }&\theta^{88 }&\theta^{354 }&\theta^{190 }&\theta^{69 }&\theta^{71 }&\theta^{338 }&\theta^{18}\\
\theta^{74 }&\theta^{87 }&\theta^{17 }&\theta^{132 }&\theta^{322 }&\theta^{324 }&\theta^{214 }&\theta^{225 }&\theta^{335 }&\theta^{187 }&\theta^{257 }&\theta^{328 }&\theta^{118 }&\theta^{95}\\
\theta^{331 }&\theta^{56 }&\theta^{28 }&\theta^{187 }&\theta^{190 }&\theta^{62 }&\theta^{319 }&\theta^{66 }&\theta^{334 }&\theta^{14 }&\theta^{269 }&\theta^{102 }&\theta^{78 }&\theta^{298}\\
\theta^{102 }&\theta^{342 }&\theta^{26 }&\theta^{227 }&\theta^{274 }&\theta^{319 }&\theta^{86}&{ 17 }&\theta^{204 }&\theta^{42 }&\theta^{125 }&\theta^{143 }&\theta^{241 }&\theta^{170}\\
\theta^{54 }&\theta^{52 }&\theta^{251 }&\theta^{164 }&\theta^{253 }&\theta^{342 }&\theta^{282 }&\theta^{266 }&\theta^{277 }&\theta^{211 }&\theta^{92 }&\theta^{298 }&\theta^{221 }&\theta^{341}\\
\theta^{207 }&\theta^{339 }&\theta^{296 }&\theta^{4 }&\theta^{165 }&\theta^{296}&{ 6 }&\theta^{77 }&\theta^{318 }&\theta^{259 }&\theta^{236}&{ 11 }&\theta^{351 }&\theta^{127}\\
\theta^{184 }&\theta^{110 }&\theta^{201 }&\theta^{27 }&\theta^{343 }&\theta^{186 }&\theta^{212 }&\theta^{53 }&\theta^{107 }&\theta^{278 }&\theta^{262 }&\theta^{2 }&\theta^{271 }&\theta^{251}\\
{8 }&\theta^{62 }&\theta^{307 }&\theta^{267 }&\theta^{341 }&\theta^{339 }&\theta^{77}&{ 5 }&\theta^{58 }&\theta^{42 }&\theta^{256 }&\theta^{3 }&\theta^{8 }&\theta^{350}\\
\theta^{81 }&\theta^{237 }&\theta^{198 }&\theta^{312}&{ 9 }&\theta^{276 }&\theta^{169 }&\theta^{124 }&\theta^{264 }&\theta^{292 }&\theta^{319 }&\theta^{296 }&\theta^{308 }&\theta^{28}\\
\theta^{357 }&\theta^{287 }&\theta^{42 }&\theta^{232 }&\theta^{234 }&\theta^{124 }&\theta^{135 }&\theta^{245 }&\theta^{97 }&\theta^{167 }&\theta^{238 }&\theta^{28 }&\theta^{270 }&\theta^{95}\\
\end{array}
\right),
$
\\ 
%[46, 12] Linear Code over GF(23^{2)
$
A_{\F_{23^2}}=
\left(
\begin{array}{lllllllllllllllll}
\theta^{119 }&\theta^{208 }&\theta^{150 }&\theta^{189 }&\theta^{476 }&\theta^{484 }&{15 }&\theta^{395 }&\theta^{326 }&\theta^{280 }&\theta^{429 }&\theta^{519 }&\theta^{56 }&\theta^{186 }&\theta^{345 }&\theta^{178 }&\theta^{522}\\
{22 }&\theta^{131 }&\theta^{503 }&\theta^{400 }&\theta^{76 }&\theta^{503 }&\theta^{155 }&\theta^{246 }&\theta^{474 }&\theta^{83 }&\theta^{407 }&\theta^{178 }&\theta^{46 }&\theta^{13}&{ \theta }&\theta^{77 }&\theta^{329}\\ 
\theta^{251 }&\theta^{490 }&\theta^{210 }&\theta^{151 }&\theta^{296 }&\theta^{221 }&\theta^{36 }&\theta^{279 }&\theta^{314 }&\theta^{37 }&\theta^{369 }&\theta^{106 }&\theta^{380 }&\theta^{313 }&\theta^{309 }&\theta^{499 }&\theta^{30}\\ 
\theta^{131 }&\theta^{62 }&\theta^{154 }&\theta^{499 }&\theta^{160 }&\theta^{26 }&\theta^{395 }&\theta^{273 }&\theta^{460 }&\theta^{518 }&\theta^{436 }&\theta^{181 }&\theta^{421 }&\theta^{232 }&\theta^{194 }&\theta^{392 }&\theta^{37}\\ 
\theta^{324 }&\theta^{91 }&\theta^{403 }&\theta^{64 }&\theta^{129 }&\theta^{39 }&\theta^{349 }&\theta^{253 }&\theta^{75 }&\theta^{285 }&\theta^{10 }&\theta^{397 }&\theta^{117 }&\theta^{422 }&\theta^{262 }&\theta^{426 }&\theta^{79}\\ 
\theta^{348 }&\theta^{150 }&\theta^{298 }&\theta^{179 }&\theta^{88 }&\theta^{402 }&\theta^{228 }&\theta^{73 }&\theta^{449 }&\theta^{294 }&\theta^{171 }&\theta^{365 }&\theta^{199 }&\theta^{512 }&\theta^{318}&{ 19 }&\theta^{507}\\ 
\theta^{213 }&\theta^{469 }&\theta^{124 }&\theta^{369 }&\theta^{498 }&\theta^{128 }&\theta^{358 }&\theta^{247 }&\theta^{36 }&\theta^{435 }&\theta^{475 }&\theta^{293 }&\theta^{462 }&\theta^{361 }&\theta^{175 }&\theta^{183 }&\theta^{208}\\ 
\theta^{405 }&\theta^{138 }&\theta^{247 }&\theta^{527 }&\theta^{492 }&\theta^{342 }&\theta^{416 }&\theta^{181 }&\theta^{14 }&\theta^{354 }&\theta^{420 }&\theta^{401 }&\theta^{194 }&\theta^{428 }&\theta^{356 }&\theta^{372 }&\theta^{363}\\ 
\theta^{396 }&\theta^{97 }&\theta^{211 }&\theta^{417 }&\theta^{417 }&\theta^{103 }&\theta^{397 }&\theta^{6 }&\theta^{243 }&\theta^{99 }&\theta^{106 }&\theta^{113 }&\theta^{69 }&\theta^{455 }&\theta^{190 }&\theta^{320 }&\theta^{319}\\ 
\theta^{228 }&\theta^{482 }&\theta^{36 }&\theta^{247 }&\theta^{173 }&\theta^{422}&{ 5 }&\theta^{381 }&\theta^{462 }&\theta^{194 }&\theta^{245 }&\theta^{193 }&\theta^{175 }&\theta^{196 }&\theta^{83 }&\theta^{20 }&\theta^{133}\\ 
\theta^{419 }&\theta^{463 }&\theta^{42 }&\theta^{221 }&\theta^{152 }&\theta^{327 }&\theta^{492 }&\theta^{157 }&\theta^{458 }&\theta^{34 }&\theta^{489 }&\theta^{481 }&\theta^{404 }&\theta^{451 }&\theta^{501 }&\theta^{62 }&\theta^{510}\\ 
\theta^{395 }&\theta^{239 }&\theta^{136 }&\theta^{340 }&\theta^{239 }&\theta^{419 }&\theta^{510 }&\theta^{210 }&\theta^{347 }&\theta^{143 }&\theta^{442 }&\theta^{310 }&\theta^{277 }&\theta^{265 }&\theta^{341 }&\theta^{65 }&\theta^{137}\\ 
\end{array}
\right),
$
\\
$
B_{\F_{23^2}}=
\left(
\begin{array}{lllllllllllllllll}
\theta^{310 }&\theta^{81 }&\theta^{54 }&\theta^{56 }&\theta^{123 }&\theta^{165 }&\theta^{148 }&\theta^{326 }&\theta^{527}&{ 8 }&\theta^{352 }&\theta^{476 }&\theta^{321 }&\theta^{414 }&\theta^{76 }&\theta^{119 }&\theta^{418}\\
\theta^{401 }&\theta^{293 }&\theta^{361 }&\theta^{517 }&\theta^{166 }&\theta^{442 }&\theta^{287 }&\theta^{107 }&\theta^{114 }&\theta^{30 }&\theta^{83 }&\theta^{47 }&\theta^{292 }&\theta^{232 }&\theta^{479 }&\theta^{251 }&\theta^{125}\\
\theta^{254 }&\theta^{309 }&\theta^{403 }&\theta^{500 }&\theta^{193 }&\theta^{345 }&\theta^{34 }&\theta^{74 }&\theta^{445 }&\theta^{396 }&\theta^{375 }&\theta^{344 }&\theta^{29 }&\theta^{522 }&\theta^{31 }&\theta^{131 }&\theta^{262}\\
\theta^{68 }&\theta^{275 }&\theta^{4 }&\theta^{127 }&\theta^{289 }&\theta^{485 }&\theta^{50 }&\theta^{462 }&\theta^{525}&{ 21 }&\theta^{326 }&\theta^{221 }&\theta^{439 }&\theta^{372 }&\theta^{434 }&\theta^{324 }&\theta^{512}\\
\theta^{224 }&\theta^{238 }&\theta^{119 }&\theta^{405 }&\theta^{65 }&\theta^{202 }&\theta^{339 }&\theta^{99 }&\theta^{6 }&\theta^{13 }&\theta^{391 }&\theta^{321 }&\theta^{465 }&\theta^{403 }&\theta^{433 }&\theta^{348 }&\theta^{383}\\
\theta^{132 }&\theta^{260 }&\theta^{476 }&\theta^{386 }&\theta^{209 }&\theta^{372 }&\theta^{450 }&\theta^{254 }&\theta^{37 }&\theta^{416 }&\theta^{486 }&\theta^{252 }&\theta^{431 }&\theta^{295 }&\theta^{330 }&\theta^{213}&{ 20}\\
\theta^{327 }&\theta^{463 }&\theta^{265 }&\theta^{510 }&\theta^{485 }&\theta^{283 }&\theta^{387 }&\theta^{132 }&\theta^{487 }&\theta^{214 }&\theta^{128 }&\theta^{114 }&\theta^{129 }&\theta^{28 }&\theta^{517 }&\theta^{405 }&\theta^{152}\\
{19 }&\theta^{462 }&\theta^{272 }&\theta^{103 }&\theta^{413 }&\theta^{363 }&\theta^{102 }&\theta^{401 }&\theta^{169 }&\theta^{468 }&\theta^{258 }&\theta^{88 }&\theta^{323 }&\theta^{58 }&\theta^{54 }&\theta^{396 }&\theta^{516}\\
\theta^{282 }&\theta^{262 }&\theta^{38 }&\theta^{405 }&\theta^{301 }&\theta^{58 }&\theta^{477 }&\theta^{411 }&\theta^{205 }&\theta^{445 }&\theta^{279 }&\theta^{513 }&\theta^{64 }&\theta^{19 }&\theta^{379 }&\theta^{228 }&\theta^{119}\\
\theta^{104 }&\theta^{50 }&\theta^{232 }&\theta^{37 }&\theta^{469 }&\theta^{340 }&\theta^{38 }&\theta^{124 }&\theta^{81 }&\theta^{347 }&\theta^{122 }&\theta^{400 }&\theta^{355 }&\theta^{154 }&\theta^{206 }&\theta^{419 }&\theta^{116}\\
\theta^{67 }&\theta^{21 }&\theta^{169 }&\theta^{380 }&\theta^{250 }&\theta^{129 }&\theta^{469 }&\theta^{362 }&\theta^{471 }&\theta^{372 }&\theta^{173 }&\theta^{392 }&\theta^{391 }&\theta^{66 }&\theta^{490 }&\theta^{395 }&\theta^{262}\\
\theta^{29 }&\theta^{97 }&\theta^{253 }&\theta^{430 }&\theta^{178 }&\theta^{23 }&\theta^{371 }&\theta^{378 }&\theta^{294 }&\theta^{347 }&\theta^{311 }&\theta^{28 }&\theta^{496 }&\theta^{215 }&\theta^{515}&{ 22 }&\theta^{521}\\
\end{array}
\right),
$
\\
%[49, 13] Linear Code over GF(5^4)
$
A_{\F_{25^2}}=\left(
\begin{array}{llllllllllllllllllll}
\theta^{509 }&\theta^{526 }&\theta^{6 }&\theta^{494 }&\theta^{286 }&\theta^{371 }&\theta^{501 }&\theta^{404 }&\theta^{274 }&\theta^{86 }&\theta^{606 }&\theta^{337 }&\theta^{181 }&\theta^{524 }&\theta^{238 }&\theta^{182 }&\theta^{130 }&\theta^{597}\\ 
\theta^{78 }&\theta^{206 }&\theta^{303 }&\theta^{116 }&\theta^{349 }&\theta^{495 }&\theta^{278 }&\theta^{152 }&\theta^{493 }&\theta^{505 }&\theta^{71 }&\theta^{441 }&\theta^{62 }&\theta^{17 }&\theta^{353 }&\theta^{16 }&\theta^{401 }&\theta^{410}\\ 
\theta^{349 }&\theta^{256 }&\theta^{471 }&\theta^{328 }&\theta^{69 }&\theta^{595 }&\theta^{162 }&\theta^{342 }&\theta^{183 }&\theta^{466 }&\theta^{562 }&\theta^{587 }&\theta^{446 }&\theta^{566 }&\theta^{516 }&\theta^{458 }&\theta^{125 }&\theta^{405}\\ 
\theta^{577 }&\theta^{444 }&\theta^{438 }&\theta^{413 }&\theta^{198 }&\theta^{232 }&\theta^{179 }&\theta^{143 }&\theta^{290 }&\theta^{73 }&\theta^{440 }&\theta^{371 }&\theta^{509 }&\theta^{243 }&\theta^{358 }&\theta^{538 }&\theta^{484 }&\theta^{46}\\ 
\theta^{351 }&\theta^{122 }&\theta^{76 }&\theta^{454 }&\theta^{357 }&\theta^{435 }&\theta^{514 }&\theta^{234 }&\theta^{165 }&\theta^{254 }&\theta^{121 }&\theta^{323 }&\theta^{367 }&\theta^{380 }&\theta^{109 }&\theta^{454 }&\theta^{14 }&\theta^{479}\\ 
\theta^{235 }&\theta^{394 }&\theta^{252 }&\theta^{590 }&\theta^{272}&{ 3 }&\theta^{591 }&\theta^{443 }&\theta^{130 }&\theta^{3  }&\theta^{176 }&\theta^{502 }&\theta^{193 }&\theta^{112 }&\theta^{120 }&\theta^{79 }&\theta^{428 }&\theta^{507}\\ 
\theta^{602 }&\theta^{433 }&\theta^{55 }&\theta^{297 }&\theta^{563 }&\theta^{538 }&\theta^{155 }&\theta^{51 }&\theta^{494 }&\theta^{123 }&\theta^{80 }&\theta^{88 }&\theta^{527 }&\theta^{93 }&\theta^{7 }&\theta^{245 }&\theta^{208 }&\theta^{452}\\ 
{3 }&\theta^{229 }&\theta^{147 }&\theta^{153 }&\theta^{323 }&\theta^{258 }&\theta^{278 }&\theta^{292 }&\theta^{155 }&\theta^{540 }&\theta^{253 }&\theta^{45 }&\theta^{166 }&\theta^{480 }&\theta^{41 }&\theta^{185 }&\theta^{427 }&\theta^{285}\\ 
\theta^{433 }&\theta^{148 }&\theta^{620 }&\theta^{298 }&\theta^{232 }&\theta^{71 }&\theta^{51}&{ 3 }&\theta^{449 }&\theta^{254 }&\theta^{99 }&\theta^{271 }&\theta^{176 }&\theta^{172 }&\theta^{481 }&\theta^{272 }&\theta^{420 }&\theta^{557}\\
\theta^{521 }&\theta^{268 }&\theta^{70 }&\theta^{302 }&\theta^{532 }&\theta^{135 }&\theta^{19 }&\theta^{396 2 }&\theta^{79 }&\theta^{592 }&\theta^{272 }&\theta^{557 }&\theta^{337 }&\theta^{328 }&\theta^{243 }&\theta^{38 }&\theta^{81}\\ 
{3 }&\theta^{230 }&\theta^{64 }&\theta^{250 }&\theta^{410 }&\theta^{309 }&\theta^{581 }&\theta^{238 }&\theta^{582 }&\theta^{284 }&\theta^{291 }&\theta^{15 }&\theta^{432 }&\theta^{592 }&\theta^{367 }&\theta^{588 }&\theta^{507 }&\theta^{197}\\ 
\theta^{525 }&\theta^{251 }&\theta^{100 }&\theta^{318 }&\theta^{432 }&\theta^{261 }&\theta^{205 }&\theta^{250 }&\theta^{498 }&\theta^{160 }&\theta^{570 }&\theta^{412 }&\theta^{249 }&\theta^{541 }&\theta^{72 }&\theta^{77 }&\theta^{302 }&\theta^{116}\\ 
\theta^{128 }&\theta^{225 }&\theta^{38 }&\theta^{271 }&\theta^{417 }&\theta^{200 }&\theta^{74 }&\theta^{415 }&\theta^{427 }&\theta^{617 }&\theta^{363 }&\theta^{608 }&\theta^{563 }&\theta^{275 }&\theta^{562 }&\theta^{323 }&\theta^{332 }&\theta^{452}\\ 
\end{array}
\right),
$
\\
$
B_{\F_{25^2}}=\left(
\begin{array}{llllllllllllllllll}
\theta^{51 }&\theta^{364 }&\theta^{572 }&\theta^{160 }&\theta^{602 }&\theta^{415 }&\theta^{103 }&\theta^{528 }&\theta^{164 }&\theta^{508 }&\theta^{170 }&\theta^{423 }&\theta^{449 }&\theta^{520 }&\theta^{260 }&\theta^{552 }&\theta^{604 }&\theta^{119}\\
\theta^{530 }&\theta^{579 }&\theta^{115 }&\theta^{523 }&\theta^{405 }&\theta^{238 }&\theta^{452 }&\theta^{376 }&\theta^{175 }&\theta^{154 }&\theta^{311 }&\theta^{139 }&\theta^{434 }&\theta^{196 }&\theta^{219 }&\theta^{155 }&\theta^{511 }&\theta^{583}\\
\theta^{25 }&\theta^{380 }&\theta^{324 }&\theta^{488 }&\theta^{502 }&\theta^{379 }&\theta^{87 }&\theta^{414 }&\theta^{173 }&\theta^{56 }&\theta^{601 }&\theta^{101 }&\theta^{326 }&\theta^{42 }&\theta^{97 }&\theta^{48 }&\theta^{368 }&\theta^{187}\\
\theta^{561 }&\theta^{416 }&\theta^{42 }&\theta^{614 }&\theta^{384 }&\theta^{393 }&\theta^{145 }&\theta^{590 }&\theta^{128 }&\theta^{595 }&\theta^{420 }&\theta^{308 }&\theta^{205 }&\theta^{475 }&\theta^{484 }&\theta^{467 }&\theta^{178 }&\theta^{585}\\
\theta^{276 }&\theta^{402 }&\theta^{152 }&\theta^{406 }&\theta^{584 }&\theta^{349 }&\theta^{233 }&\theta^{98 }&\theta^{378}&{ 1 }&\theta^{409 }&\theta^{201 }&\theta^{486 }&\theta^{428 }&\theta^{367 }&\theta^{304 }&\theta^{47 }&\theta^{469}\\
\theta^{583 }&\theta^{615 }&\theta^{12 }&\theta^{390 }&\theta^{250 }&\theta^{423 }&\theta^{63 }&\theta^{60 }&\theta^{384 }&\theta^{124}&{ 4 }&\theta^{64 }&\theta^{253 }&\theta^{583 }&\theta^{194 }&\theta^{61 }&\theta^{382 }&\theta^{212}\\
\theta^{142 }&\theta^{453 }&\theta^{380 }&\theta^{405 }&\theta^{389 }&\theta^{244 }&\theta^{292 }&\theta^{45 }&\theta^{501 }&\theta^{285 }&\theta^{591 }&\theta^{122 }&\theta^{271 }&\theta^{505 }&\theta^{504 }&\theta^{43 }&\theta^{294 }&\theta^{78}\\
\theta^{140 }&\theta^{65 }&\theta^{271 }&\theta^{202 }&\theta^{457 }&\theta^{436 }&\theta^{166 }&\theta^{327 }&\theta^{539 }&\theta^{455 }&\theta^{181 }&\theta^{454 }&\theta^{382 }&\theta^{576 }&\theta^{479 }&\theta^{406 }&\theta^{329 }&\theta^{43}\\
\theta^{26 }&\theta^{116 }&\theta^{560 }&\theta^{146 }&\theta^{307 }&\theta^{557 }&\theta^{411 }&\theta^{254 }&\theta^{250 }&\theta^{546 }&\theta^{404 }&\theta^{97 }&\theta^{143 }&\theta^{116 }&\theta^{603 }&\theta^{434 }&\theta^{121 }&\theta^{131}\\
\theta^{453 }&\theta^{157 }&\theta^{142 }&\theta^{590 }&\theta^{406 }&\theta^{562 }&\theta^{63 }&\theta^{30 }&\theta^{332 }&\theta^{412 }&\theta^{26 }&\theta^{475 }&\theta^{565 }&\theta^{32 }&\theta^{298 }&\theta^{89 }&\theta^{304 }&\theta^{78}\\
\theta^{475 }&\theta^{458 }&\theta^{57 }&\theta^{46 }&\theta^{100 }&\theta^{535 }&\theta^{566 }&\theta^{180 }&\theta^{606 }&\theta^{368 }&\theta^{390 }&\theta^{595 }&\theta^{193 }&\theta^{328 }&\theta^{88 }&\theta^{282 }&\theta^{457 }&\theta^{135}\\
\theta^{41 }&\theta^{554 }&\theta^{432 }&\theta^{35 }&\theta^{254 }&\theta^{303 }&\theta^{613 }&\theta^{133 }&\theta^{206 }&\theta^{92 }&\theta^{420 }&\theta^{409 }&\theta^{387 }&\theta^{30 }&\theta^{458 }&\theta^{146 }&\theta^{100 }&\theta^{362}\\
\theta^{501 }&\theta^{37 }&\theta^{445 }&\theta^{327 }&\theta^{160 }&\theta^{374 }&\theta^{298 }&\theta^{97 }&\theta^{76 }&\theta^{233 }&\theta^{61 }&\theta^{356 }&\theta^{118 }&\theta^{141 }&\theta^{77 }&\theta^{433 }&\theta^{505 }&\theta^{546}\\
\end{array}
\right).
$
}
\end{table*}

%
%\begin{IEEEbiographynophoto}
%{Lin Sok} received the B.Sc. degree in Mathematics from the Royal University of Phnom Penh, Phnom Penh, Cambodia, in 2003, the M.Sc. degree in Discrete Mathematics and Fundamental of Computer Sciences from the Universit\'e de la M\'editerran\'ee-Aix Marseille II, Marseille, France, in 2008 and the Ph.D. degree in Communications and Electronic from Telecom ParisTech, Paris, France, in 2013. From 2013 to 2014, he was a postdoctoral fellow at the Department of Communications and Electronic, Telecom ParisTech. He was a postdoctoral fellow at School of Mathematical Sciences, Anhui University, China from 2016 to 2018 and became an assistant professor since 2018. He was an assistant professor at the Department of Mathematics, Royal University of Phnom Penh, Cambodia from 2014 to 2018. His research interests include Algebraic Coding Theory (codes over rings, algebraic geometry codes, lattice codes) and Cryptography (Boolean functions).
\end{document}